\documentclass[format=sigconf,letterpaper,9pt]{acmart}

\usepackage{balance} 
\usepackage[inline]{enumitem}
\usepackage{makecell}
\usepackage{graphicx}
\usepackage[title]{appendix}
\usepackage{threeparttable}
\usepackage{booktabs}
\usepackage{libertine}

\theoremstyle{definition}
\newtheorem{experiment}{Experiment}[subsection]


\begin{document}

\title[Fairness Under Unawareness]{Fairness Under Unawareness:\\Assessing Disparity When Protected Class Is Unobserved}

\newcommand{\pr}{\mathbb{P}}
\newcommand{\expect}{\mathbb{E}}
\newcommand{\expectn}{\hat{\mathbb{E}}_N}
\newcommand{\var}{\operatorname{Var}}
\newcommand{\cov}{\operatorname{Cov}}
\newcommand{\na}{\text{NA}}
\newcommand{\ind}{\mathbb{I}}
\newcommand{\plim}{\operatorname{plim}}

\hyphenation{di-s-cr-i-mi-na-ti-on}
\hyphenation{ge-o-lo-ca-ti-on}
\hyphenation{in-t-ra-ge-o-lo-ca-ti-on}
\hyphenation{in-te-r-ge-o-lo-ca-ti-on}
\hyphenation{thr-esh-old}
\hyphenation{gr-o-up}
\hyphenation{hi-gh-er}
\hyphenation{thr-esh-old-ed}
\hyphenation{w-ei-ght-ed}

\setlength{\emergencystretch}{1em}


\author{Jiahao Chen}
\affiliation{}
\email{cjiahao@gmail.com}
\orcid{0000-0002-4357-6574}

\author{Nathan Kallus}
\affiliation{%
    \institution{Cornell Tech}
    \streetaddress{2 West Loop Road}
    \city{New York}
    \state{New York}
    \postcode{10044-1501}
    \country{USA}
}
\email{kallus@cornell.edu}

\author{Xiaojie Mao}
\authornote{Corresponding author}
\affiliation{%
    \institution{Cornell Tech}
    \streetaddress{2 West Loop Road}
    \city{New York}
    \state{New York}
    \postcode{10044-1501}
    \country{USA}
}
\email{xm77@cornell.edu}

\author{Geoffry Svacha}
\affiliation{}
\email{svacha@gmail.com}

\author{Madeleine Udell}
\affiliation{%
    \institution{Cornell University}
    \streetaddress{223 Rhodes Hall}
    \city{Ithaca}
    \state{New York}
    \postcode{14853-3801}
    \country{USA}
}
\email{udell@cornell.edu}

\renewcommand{\shortauthors}{J. Chen, N. Kallus, X. Mao, G. Svacha, M. Udell}

\begin{abstract}
Assessing the fairness of a decision making system with respect to a protected
class, such as gender or race, is challenging when class membership labels
are unavailable.
Probabilistic models for predicting the protected class based on observable proxies,
such as surname and geolocation for race, are sometimes used to
impute these missing labels for compliance assessments.
Empirically, these methods are observed to exaggerate disparities, but the
reason why is unknown.
In this paper, we decompose the biases in estimating outcome disparity via
threshold-based imputation into multiple interpretable bias sources,
allowing us to explain when over- or underestimation occurs.
We also propose an alternative weighted estimator that uses soft classification,
and show that its bias arises simply from the conditional covariance of
the outcome with the true class membership.
Finally, we illustrate our results with numerical simulations and a public dataset 
of mortgage applications, using geolocation as a proxy for race.
We confirm that the bias of threshold-based imputation is generally upward, but its
magnitude varies strongly with the threshold chosen.
Our new weighted estimator tends to have a negative bias that is much simpler
to analyze and reason about.
\end{abstract}

\copyrightyear{2019}
\acmYear{2019}
\setcopyright{acmcopyright}
\acmConference[FAT* '19]{FAT* '19: Conference on Fairness, Accountability, and
	Transparency}{January 29--31, 2019}{Atlanta, GA, USA}
\acmBooktitle{FAT* '19: Conference on Fairness, Accountability, and Transparency
	(FAT* '19), January 29--31, 2019, Atlanta, GA, USA}
\acmPrice{15.00}
\acmDOI{10.1145/3287560.3287594}
\acmISBN{978-1-4503-6125-5/19/01}


\maketitle

\section{Introduction}

Models for high stakes decision making have ethical and legal needs to demonstrate a lack of discrimination with respect to protected classes \cite{Munoz2016,Barocas:2016aa}.
Examples of such decisions include employment and compensation \cite{Conway1983,Greene1984}, university admissions \cite{Bickel1975}, and sentence and bail setting \cite{Berk2017,Chouldechova2017,Dressel2018}.
Another example relevant to the financial services industry is credit decisioning \cite{Chen2018}, which is a classification problem where these ethical concerns are enshrined in concrete regulatory compliance requirements.
Credit decisions must be shown to comply with a myriad of federal and state fair lending laws, some of which are summarized in \cite{Chen2018}%
\footnote{In this paper, we restrict our discussion and citation of applicable laws to those of the United States of America.}.
Some of these laws define protected classes, such as race and gender,
where discrimination on the basis of a customer's membership in these classes is prohibited. Table \ref{table:classes} summarizes the protected classes defined by the Fair Housing Act (FHA) \cite{fha} and Equal Credit Opportunity Act (ECOA) \cite{ecoa}. 




\begin{table}[h]

\caption{\label{table:classes} Protected classes defined under fair lending laws.}

\begin{tabular}{l c c}
\toprule
Law & FHA\cite{fha}\! & ECOA\cite{ecoa}\! \\
\midrule
age & & X \\
color & X & X \\
disability & X & \\
exercised rights under CCPA & & X \\
familial status (household composition)\!\!& X & \\
gender identity & X & \\
marital status (single or married)& & X \\
national origin & X & X \\
race  & X & X \\
recipient of public assistance & & X \\
religion & X & X \\
sex & X & X \\
\bottomrule
\end{tabular}

\end{table}

\subsection{Assessing fairness with unknown protected class membership}

When demonstrating that credit decisions comply with these fair lending laws,
we sometimes run into situations where fairness and bias assessments must be done on populations without knowing their memberships in protected classes, because it is illegal or operationally difficult to do so.
For example, credit card and auto loan companies must demonstrate that the way they extend credit is not racially discriminatory, yet are not allowed to ask applicants what race they are when they apply for credit.\footnote{Lenders may ask applicants to self-identify in a voluntary basis, with the understanding that the answer will not affect the outcome of the application and that the information is collected for compliance assessment only \cite[12 CFR \S 1002.5(b)]{regulationb}.}
Similarly, health plans can only solicit race and ethnicity information for new members but cannot obtain the same information for existing members \cite{elliott2008new}.
Given the lack of secure protocols that permit disparity evaluation with encrypted protected classes \cite{kilbertus2018blind}, disparate impact assessments for these situations have to impute the mostly (or entirely) missing labels corresponding to the protected class, usually by relying on observed proxy variables that can predict class memberships.
 The imputed protected classes are then used by regulators in assessing disparate impact (but they are not allowed to be used in decision making).
 Generally, any model that imputes the missing protected attribute value based on other, observed variables is known as a \emph{proxy model}, and such a model that is based on predicting conditional class membership probabilities is known as a \emph{probabilistic proxy model}.


For example, for assessing adverse action with regard to race in credit decisions, regulators like the Bureau of Consumer Financial Protection (BCFP)\footnote{Formerly the Consumer Financial Protection Bureau (CFPB). Citations and references reflect the name at time of publication.}  have been known in the past to use a probabilistic proxy model to impute the customers' unknown race labels \cite{CFPB:proxy2014}.
They used a na\"ive Bayes classifier, the Bayesian Improved Surname Geocoding (BISG) method, to predict the probability of race membership given the customer's surname and address of residence \cite{Elliott:2009aa}.
Specifically, assuming that surname and location are statistically independent given race, BISG uses Bayes's rule to
compute race membership probabilities from
the conditional distributions of surname given race and of location given race as inferred by census data.
This methodology \cite{CFPB:proxy2014} notably supported a \$98 million fine against a major auto loan lender \cite{CFPB:ally2013}.
This case generated some controversy \cite{Koren:2016aa,CFS:2015,CFS:2016,CFS:2017,CFPB:2018-05}, in part due to empirical findings that the amount of disparate impact estimated by BISG appears to overestimate true disparities \cite{Baines:2014aa,zhang2016assessing}.
However, the cause for this overestimation phenomenon is unknown, as is whether overestimation is to be expected always, or whether or not underestimation of disparate impact is also possible.
This observation forms the motivation for our current work, which is broadly applicable to any fairness assessment where an unobserved protected class must be imputed using a proxy model.
The aim is not to criticize the use of proxy models in general, but rather to provide a more informed analysis of the statistical biases inherent in any assessment where membership in protected classes must be imputed.

\paragraph{Main results}

This paper investigates the bias in estimating demographic disparity (Definition \ref{demographic-disparity}) when a proxy model is employed to impute a protected class.
We present the first theoretical results describing when the use of proxy models can lead to biased estimates of outcome disparity,
which explains the overestimates observed in the past, and also offers insights on the practical use of proxy models.
More specifically, our key contributions are:

\begin{enumerate}
\item
We derive the (statistical) bias for the commonly used thresholded estimator, where a label is assigned only if the proxy model predicts a label with probability exceeding a predefined threshold \cite{CFPB:proxy2014,Baines:2014aa,zhang2016assessing} (Definition \ref{def:thresholded_estimator}, Theorem \ref{thm: demographic-est-race-bias}). We decompose its bias into multiple sources, which gives a set of interpretable conditions under which the thresholded estimator can over- or underestimate the outcome disparity.

\item
We present a new weighted estimator for demographic disparity (Definition \ref{def:weighted_estimator}) that uses soft classification based on proxy model outputs as opposed to hard imputation. We derive its bias (Theorem \ref{thm: weighted}) and find that the weighted estimator has only one bias source.


\item
We validate our results on a public mortgage data set, using geolocation as the sole variable in a proxy model for race.
We identify the specific source of bias that can account for the overestimation of the thresholded estimator, which can explain the overestimation bias of using proxy methods observed in previous literature.
\item
We discover that the estimation bias is sensitive to the threshold used in class imputation based on the proxy model. This shows the intrinsic limitation of the thresholded estimator.
\end{enumerate}

\section{Evaluating the fairness of a binary decision}

We have three main variables of interest:

\begin{description}
\item[Binary decision $Y$]
with $Y = 1$ representing a favorable outcome, such as the approval of loan application or college admission offer, and $Y = 0$ representing an unfavorable outcome.

\item[Protected class $A$]
such as gender or race; often, we will 
write $A=a$ for
the \textit{advantaged group} and $A=b$ for the
\textit{disadvantaged group}.

\item[Proxy variable $Z$] a set of covariates taking values $z\in\mathcal{Z}$ used to predict ${A}$ in a probabilistic proxy model.
\end{description}

We present only the binary case $A\in\{a,b\}$ for simplicity. Unless otherwise stated, for multiclass $A$, our results generalize straightforwardly to the pairwise outcome disparity between any advantaged group and any disadvantaged group. 
(Additional details are given in Appendices B and C.)


\begin{definition}
\label{mean-group-outcome}
The \textit{mean group outcome} for the group $A=u$ is
\[
\mu(u) = \expect (Y \mid A = u),
\]
where $u\in\{a,b\}$ and $\expect$ is the usual expectation with respect to population distribution.
\end{definition}

\begin{definition}
\label{demographic-disparity}
The \textit{demographic disparity}, or Calders-Verwer gap \cite{Calders2010,Kamishima2012}, $\delta$, is the difference in mean group outcomes between the advantaged and disadvantaged groups:
\begin{equation}\label{def:demo-disparity}
\delta = \mu(a) - \mu(b)
\end{equation}
\end{definition}
A positive demographic disparity $\delta$ means that a higher proportion of the advantaged group $A=a$ receive a favorable outcome than the disadvantaged group $b$, i.e. the disadvantaged group experiences an outcome disparity.
Demographic disparity is simple to understand and is widely used \cite{lipton2018does, zliobaite2015relation, zafar2016learning}, despite its flaws \cite{dwork2012fairness, hardt2016equality}.

When the labels for the protected class $A$ are known, the demographic disparity can be simply and reliably estimated by the difference of within-group sample means.
Otherwise, a probabilistic proxy model may be used to estimate the probabilities $\pr (A = a \mid Z)$ and $\pr (A = b \mid Z)$ of membership to the different groups within $A$. BISG \cite{Elliott:2009aa} is an example of a probabilistic proxy model where $Z$ is taken to be surname and geolocation and conditional probabilities are estimated using the na\"ive Bayes methodology.
%
%
The goal of this paper is to quantify the bias in estimating demographic disparity using a proxy model.
We neglect the estimation bias intrinsic to estimating the proxy model itself, so that $\{\pr (A = u \mid Z) : u\in\{a,b\} \}$ describes the true population distribution of $A$ conditioned on $Z$ and hence the intrinsic uncertainty of predicting $A$ from $Z$.

\subsection{Thresholded estimator}

In order to introduce the estimators for demographic disparity, assume that we have $N$ independent and identically distributed (iid) samples $(Y_i, Z_i)_{i=1}^N$. The true membership $A_i$ of the $i$th sample is unknown, but we have access to the probabilistic proxy estimates $\{\pr (A_i = u \mid Z_i) : u\in\{a,b\}, i\in\{1,\dots,N\}\}$ by simply applying our proxy model to the observed proxy variable $Z_i$.

The ordinary approach is to predict a single value for class membership, $\hat{A}_i$ \cite{Baines:2014aa, zhang2016assessing}:

\begin{definition}
\label{def:thresholding_rule}
Let $q \in [\frac{1}{2}, 1)$.
Then the \textit{thresholded estimated membership} $\hat{A}_i$ for the $i$th unit is
\begin{equation*}
\hat{A}_i =\left\{
\begin{aligned}
&a, \ \pr (A_i = a \mid Z_i) > q, \\
&b, \ \pr (A_i = b \mid Z_i) > q, \\
&\text{NA}, \ \text{otherwise},
\end{aligned}
\right.
\end{equation*}
where NA stand for an unclassified observation that is excluded from the subsequent outcome disparity evaluation.
\end{definition}


Considering a unit with $A=a$, this estimation rule can be summarized pictorially as follows
\begin{center}
\includegraphics[width=0.6\linewidth]{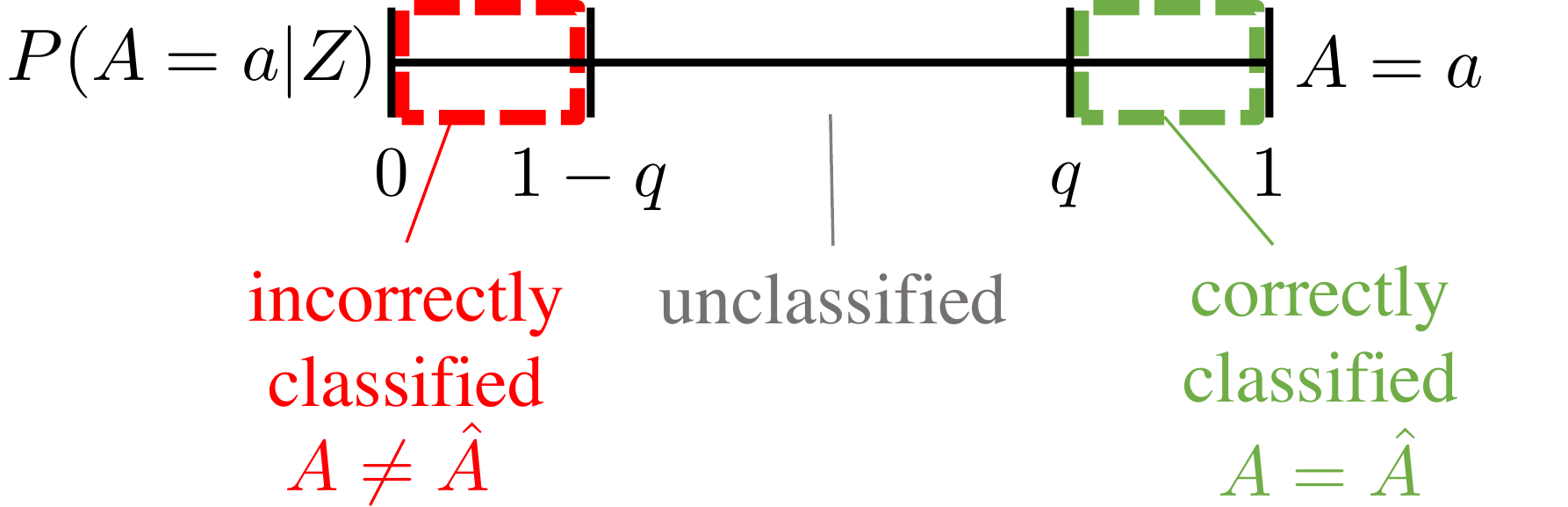}
\end{center}
where dashed boxes represent correct ($\hat{A}=A$, green) and incorrect classifications ($\hat{A}\ne A$, red), and the middle is unclassified.

We can use these predicted labels to estimate the mean group outcomes 
and demographic disparity by mimicking the simple estimator for 
group sample means difference, but imputing $\hat A$ for the
unknown $A$. This leads to the following estimator:

\begin{definition}
\label{def:thresholded_estimator}
Let $(Y_i, Z_i)_{i=1}^N$ be $N$ iid samples, $\{\hat{A}_i\}_{i=1}^N$ be the estimated labels according to the thresholding rule in Definition \ref{def:thresholding_rule}, and $\ind(S)$ be the indicator function for some set $S$. Then, the \textit{thresholded estimators} for mean group outcomes and demographic disparity are
\begin{align}
\hat{\mu}_{q}(u) & = \frac{\sum_{i = 1}^N \mathbb{I}(\hat{A}_i  = u)Y_i}{\sum_{i = 1}^N \mathbb{I}(\hat{A}_i  = u)}, u\in\{a,b\},  \nonumber \\
\hat{\delta}_{q} & = \hat{\mu}_{q}(a) - \hat{\mu}_{q}(b). \label{demo: est_race}
\end{align}
\end{definition}




\subsection{Weighted estimator}

The form of the thresholded estimators above shows clearly the use of a hard classification rule.
Under a probabilistic proxy model, this hard classification rule inevitably misclassifies some individuals,  given the intrinsic uncertainty in classifying the protected class.
Moreover, the threshold rule results in a group of unclassified individuals who are removed from the outcome disparity evaluation. 
To avoid these problems, we propose a new estimator that accounts for the soft classification generated by the proxy.

\begin{definition}
\label{def:weighted_estimator}
Let $(Y_i, Z_i)_{i=1}^N$ be defined as above.
Then, the \emph{weighted estimators} for mean group outcomes and demographic disparity are

\begin{align}
\hat{\mu}_W(u) & = \frac{\sum_{i = 1}^N \pr(A_i = a \mid Z_i)Y_i}{\sum_{i = 1}^N \pr(A_i = a \mid Z_i)}, u\in\{a,b\}  \nonumber \\
\hat{\delta}_W & = \hat{\mu}_W(a) - \hat{\mu}_W(b) \label{demo: weight}
\end{align}
\end{definition}



\subsection{Toy examples}
\label{sec:toy_examples}

In this part, we use two hypothetical toy examples to demonstrate the intuition regarding when the thresholded estimator \eqref{demo: est_race} can overestimate the demographic disparity. In both examples, we consider evaluating the disparity of loan approval with respect to two races $A = a$ and $A = b$. However, the true race $A$ is unknown, so the geolocation $Z$ is used to estimate race. Suppose that there are only two neighborhoods where all people live in one neighborhood have high income and all people live in the other neighborhood have low income. We further assume that the high-income neighborhood is primarily occupied by the advantaged group $a$ and the low-income neighborhood is primarily occupied by the disadvantaged group $b$. Therefore, the thresholding rule with $q = 0.5$ classifies all people living in the high-income neighborhood as the advantaged group, i.e., $\hat{A} = a$,  and all people living in the low-income neighborhood as the disadvantaged group, i.e., $\hat{A} = b$.

\begin{table}
\centering
\caption{Lending policy of Example \ref{ex:income}, which always accepts high
income applicants and never low income applicants.
Each number in the table cells
represents the number of applicants who truly belong to race $a$ or $b$.
The percentages in parentheses represent the average loan approval
rate in that quadrant.} \label{table:toy1}
\begin{tabular}{c c c}\toprule
& \multicolumn{2}{c}{True race} \\
Neighborhood & $a$ & $b$ \\
\midrule
 high income & 70 (100\%)          & 30 (100\%)\textsuperscript{i} \\
 low income  & 30 (0\%)\textsuperscript{ii} & 70 (0\%) \\
\bottomrule
\end{tabular}\\
\textsuperscript{i} Misclassified as race $a$ in high-income neighborhood.\\
\textsuperscript{ii} Misclassified as race $b$ in low-income neighborhood.
\end{table}

\begin{example}
\label{ex:income}
Consider an extreme lending policy: all people with high income get their loans approved, while all people with low income are rejected, no matter what their races are.
See Table \ref{table:toy1} for the illustration.
  Simple calculations show that the true demographic disparity $\delta = 40\%$ \footnote{$\delta = \frac{1}{100}(70\times 1 + 30\times 0) - \frac{1}{100}(30\times 1 + 70\times 0)  = 0.4$.} while the thresholded estimator $\hat{\delta}_{0.5} = 100\%$ \footnote{$\hat{\delta}_{0.5} = \frac{1}{100}(70\times 1 + 30\times 1) - \frac{1}{100}(30\times 0 + 70\times 0)  = 1$.}. Thus the thresholded estimator overestimates the demographic disparity. The reason for the overestimation is that the race proxy is correlated with the loan approval outcome because of the dependence between geolocation and socioeconomic status: people who live in the neighborhood primarily occupied by the advantaged group are also more likely to get loan approval because they have relatively high socioeconomic status, e.g.\ high income in this example. These people are classified as advantaged group because their neighborhoods are associated with high probability of belonging to the advantaged group. As a result, the thresholding rule misclassifies people who are from the disadvantaged group but get loan approval as the advantaged group. In this way, the thresholding rule leads to underestimates of the loan acceptance rate of the disadvantaged group. Analogously, the thresholding rule leads to overestimates of the loan acceptance rate of the advantaged group because it misclassifies people from the advantaged group but not likely to get loan approval as the disadvantaged group. Consequently, the misclassification of the protected attribute is dependent with the outcome, which ultimately leads to the overestimation of demographic disparity. The dependence between the misclassification and the outcome results from the \textit{inter-geolocation} outcome variation: the socioeconomic status, race proxy probabilities (i.e., race ratios), and loan acceptance rates vary across different neighborhoods such that the favorable outcome is positively correlated with the probability of belonging to the advantaged group.
\end{example}

\begin{table}
\centering
\caption{Lending policy of Example \ref{ex:affirmative_action}, with
affirmative action that favors the disadvantaged group over the advantaged
group at any given income level.\label{table:toy2}}
\begin{tabular}{c c c}\toprule
& \multicolumn{2}{c}{True race} \\
Neighborhood & $a$ & $b$ \\
\midrule
high income & 70 (70\%) & 30 (80\%)\textsuperscript{i}\\ 
low income & 30 (20\%)\textsuperscript{ii} & 70 (30\%)\\ 
\bottomrule
\end{tabular}\\
\textsuperscript{i} Misclassified as race $a$ in high-income neighborhood.\\
\textsuperscript{ii} Misclassified as race $b$ in low-income neighborhood.
\end{table}

\begin{example}
\label{ex:affirmative_action}
  Now consider a hypothetical lending policy with affirmative action that approves the disadvantaged group with higher rate than the advantaged group with the same income level. But overall people with high income are still more likely to be accepted. See Table \ref{table:toy2} for a concrete example. Simple calculations show that the true demographic disparity $\delta = 10\%$ \footnote{$\delta = \frac{1}{100}(70\times 0.7 + 30 \times 0.2) - \frac{1}{100}(30\times 0.8 + 70 \times 0.3) = 0.1$.}, which means that the disadvantaged group overall has lower chance to get loan approval due to their population concentration in the low-income neighborhood. However, the thresholded estimator gives  $\hat{\delta}_{0.5} = 46\%$ \footnote{$\hat{\delta}_{0.5} = \frac{1}{100}(70\times 0.7 + 30 \times 0.8) - \frac{1}{100}(30\times 0.2 + 70 \times 0.3) = 0.46$.} that overestimates the demographic disparity. The reason for the overestimation is that the the disadvantaged group have higher approval rate than their neighbors from the advantaged group. As a result, misclassifying part of the disadvantaged group as the advantaged group raises the average approval rate for the advantaged group. Similarly, misclassifying part of the advantaged group as the disadvantaged group lowers down the average approval rate for the disadvantaged group. Consequently, the misclassification of the protected attribute is also dependent with the outcome and ultimately leads to overestimation of demographic disparity. In contrast to Example \ref{ex:income}, the dependence between the misclassification and the outcome results from the intra-geolocation outcome variation: people from different protected groups living in the same locations have different chance of getting the favorable outcome.
\end{example}

In both examples, the protected attribute misclassification made by the thresholding rule is not uniformly at random. Instead, the misclassification shows systematic pattern with respect to the outcome either due to the inter-geolocation outcome variation (Example \ref{ex:income}) or intra-geolocation outcome variation (Example \ref{ex:affirmative_action}). In Section \ref{sec:theorems}, we will formalize the main intuition in these two examples and show that the interplay of the two bias sources captured by these two examples determines the overestimation or underestimation of the thresholded estimator. In Section \ref{sec:numerics}, we will further identify the specific bias source that accounts for the estimation bias of the thresholded estimator in a mortgage data set.

\section{Bias in thresholded and weighted estimators}
\label{sec:theorems}

In this section, we derive the asymptotic biases for the thresholded estimator \eqref{demo: est_race} and weighted estimator \eqref{demo: weight} for demographic disparity. We also provide some interpretable sufficient conditions under which these two estimators overestimate or underestimate the demographic disparity.


\subsection{Weighted estimator}

\begin{theorem}\label{thm: weighted}
Let $A$ be a binary protected class with values $a$ and $b$. The bias of the weighted estimator $\hat{\delta}_{W}$ in Definition \ref{def:weighted_estimator} for demographic disparity $\delta$ in \eqref{def:demo-disparity} is
\[
\hat{\delta}_{W} - \delta = [\hat{\mu}_{W}(a) - \mu(a)] -  [\hat{\mu}_{W}(b) - \mu(b)],
\]
where as $N \to \infty$, the biases in the weighted estimators for the mean group outcomes $\hat{\mu}_{W}(u)$, for $u\in\{a,b\}$, converge almost surely to
\begin{equation}
\hat{\mu}_{W}(u) - \mu(u) \xrightarrow{\textrm{a.s.}} -\frac{\expect [\cov (\ind(A = u), Y \mid Z)]}{\pr (A = u)}.\label{formula: weight-bias1}
\end{equation}

\end{theorem}

We omit the $a.s.$ (almost sure) notation later for brevity.

\begin{corollary}
\label{corr:weighted-unbiased}
If $Y$ is independent of $A$ conditionally on $Z$, then the weighted estimator for demographic disparity, $\hat{\delta}_{W}$, is asymptotically unbiased.
\end{corollary}

If $Y$ is independent of $A$ conditionally on $Z$, then the advantaged group and disadvantaged group with the same $Z$ values are equally treated in terms of the average outcome.
Nevertheless, this does not contradict the existence of overall disparity against the disadvantaged group in terms of the unconditioned average outcome, as an example of Simpson's paradox \cite{Simpson1951}.
The conditional independence assumption required by Corollary \ref{corr:weighted-unbiased}
is trivially satisfied if $Y$ is the output of some function $f(X)$, and $Z$ includes the input features $X$, since $Y$ is now determined entirely by $Z$.
Such a situation arises naturally when $Y$ is the output of machine learning algorithms. In the Appendix \ref{section: unbiased_weighted}, we show a semi-synthetic example based on the mortgage dataset where the weighted estimator is asymptotically unbiased when the proxy model is well constructed so that the conditional independence assumption is satisfied.

However, this conditional independence assumption may not hold in practice and the weighted estimator may be biased. Consider the example of loan application with race proxy based on geolocation. If within most locations, the affirmative action described in Example \ref{ex:affirmative_action} is present, i.e., the disadvantaged group is more likely to get a loan than their neighbors from the advantaged group, then $Y$ covaries negatively with $\ind(A = a)$ and positively with $\ind(A = b)$ conditionally on $Z = z$, for most values of $z \in \mathcal{Z}$, which implies that the weighted estimator overestimates the demographic disparity. On the contrary, if within most locations, the advantaged group is more likely to get loan than their neighbors belonging to the disadvantaged group, then $Y$ covaries in the exact opposite way, implying that the weighted estimator underestimates the demographic disparity. Overall, the estimation bias of the weighted estimator depends on the intra-geolocation variation of the loan application outcome.

\subsection{Thresholded estimator}

Before stating the asymptotic bias for the estimator \eqref{demo: est_race}, we define the following terms for $u \in \{a, b\}$:
\begin{align*}
    \Delta_1(u) = &\; \expect [Y \mid \pr(A = u \mid Z) > q, A = u^c]\\
    & - \expect [Y \mid \pr(A = u \mid Z) > q, A = u], \\
    \Delta_2(u) = &\; \expect [Y \mid \pr(A = u\mid Z) \le q, A = u] \\
    & - \expect [Y \mid \pr(A = u\mid Z) > q, A = u],
\end{align*}
where $q$ is the threshold used for estimating race, and $u^c$ is the class opposite to $u$, i.e., $u^c = b$ if $u = a$ and $u^c = a$ if $u = b$. $\Delta_1(u)$ measures the outcome mean discrepancy for two different protected groups \textit{within} the same proxy probability range, and $\Delta_2(u)$ measures the outcome mean discrepancy \textit{across} different proxy probability ranges for the same protected group. Consider the example of loan application with race proxy based on geolocation. In this example, $\Delta_1(u)$ measures the loan approval disparity between two race groups who live in locations that are primarily occupied by one of these two races (in terms of the threshold $q$). In contrast $\Delta_2(u)$ measures the loan approval rate disparity between people belonging to a race who live in locations that are primarily occupied by this race group and people belonging to this race who live in locations that are less occupied by this race. Therefore, $\Delta_1(u)$ roughly characterizes the \textit{intra-geolocation} variation of loan outcome and $\Delta_2(u)$ roughly characterizes the \textit{inter-geolocation} variation of loan outcome.

\begin{theorem} \label{thm: demographic-est-race-bias}
Let $A$ be a binary protected class with values $a$ and $b$. 
The bias for the thresholded estimator $\hat{\delta}_q$ in Definition \ref{def:thresholded_estimator} is:
\[
\hat{\delta}_{q} - \delta = [\hat{\mu}_{q}(a) - \mu(a)] -  [\hat{\mu}_{q}(b) - \mu(b)],
\]
and as ${N \to \infty}$, for $u \in \{a, b\}$ and $u^c \in \{a, b\}$ as the class opposite to $u$, 
\begin{align*}
    \hat{\mu}_{q}(u) - \mu(u) \to &\;\Delta_1(u)C_1(u) - \Delta_2(u)C_2(u) \\
    &+ (\Delta_1(u) - \Delta_2(u))C_3(u),
\end{align*}
where
\begin{align*}
    C_1(u) & = \pr(\hat{A} = u \mid A = u)\pr(A = u^c \mid \hat{A} = u), \\
    C_2(u) & = \pr(A = u \mid \hat{A} = u)\pr(\hat{A} \ne u \mid A = u)\textrm{, and}\\
    C_3(u) & = \pr(\hat{A} \ne u \mid A = u)\pr(A = u^c \mid \hat{A} = u).
\end{align*}
Here $\hat{A} \ne u$ means that $\hat{A} = u^c$ or $\hat{A} = \text{NA}$.
\end{theorem}
The generalization to multiclass $A$ is straightforward and in Appendix \ref{section:theory_bias} we show that the bias formula in Theorem \ref{thm: demographic-est-race-bias} for binary protected class capture the main effects of the bias for multiclass $A$.

Theorem \ref{thm: demographic-est-race-bias} shows that the occurrence of overestimation or underestimation depends on complex interplay of the $\Delta$ terms. In the following corollary, we provide the simplest set of sufficient conditions for the overestimation and underestimation to demonstrate the main intuition.

\begin{corollary} \label{corollary: demographic-est-race}
Suppose
\begin{subequations}
\begin{align}
  \Delta_1(a) &\ge 0,\tag{5i}\label{eq:condition1} \\
  -\Delta_1(b) &\ge 0,\tag{5ii}\label{eq:condition2} \\
 -\Delta_2(a) &\ge 0,\tag{5iii}\textrm{ and} \label{eq:condition3}\\
  \Delta_2(b) &\ge 0,\tag{5iv}\label{eq:condition4}
\end{align}
\end{subequations}
and at least one inequality holds strictly.
Then in the limit $N\!\rightarrow\!\infty$,
$\hat{\mu}_q(a) > \mu(a)$, $\hat{\mu}_q(b) < \mu(b)$, and $\hat{\delta}_q > \delta$, i.e. the thresholded estimator overestimates the demographic disparity.

Conversely, let \begin{enumerate*}[label=(\roman*\ensuremath{^\prime})]
\item $\Delta_1(a)\!\le\!0$,
\item $\,-\Delta_1(b)\!\le\!0$,
\item $\,-\Delta_2(a)\!\le\!0$, and
\item $\Delta_2(b)\!\le\!0$,
\end{enumerate*}
and at least one inequality holds strictly.
Then in the limit $N\rightarrow\infty$,
$\hat{\mu}_q(a) < \mu(a)$, $\hat{\mu}_q(b) > \mu(b)$, and $\hat{\delta}_q < \delta$, i.e. the thresholded estimator underestimates the demographic disparity.
\end{corollary}

The conditions \eqref{eq:condition1}---\eqref{eq:condition4} are demonstrated in Figure \ref{figure: theorem}.

Conditions \eqref{eq:condition1}---\eqref{eq:condition2} exactly formalize the intuition captured by Example \ref{ex:affirmative_action}.
In our loan application example, using a geolocation-based race proxy, these two conditions capture the intra-geolocati\-on outcome variation:
on average higher proportion of disadvantaged group $b$ receive the favorable outcome than the advantaged group $a$ among all the geolocations that are primarily occupied by one race (in terms of the threshold $q$).

The conditions \eqref{eq:condition3}---\eqref{eq:condition4} exactly formalize the intuition captured by Example \ref{ex:income}.
They characterize the dependence between the decision outcome and the proxy probability conditionally on the protected attribute.
Specifically, condition \eqref{eq:condition3} holds when the decision outcome is positively correlated with the probability of belonging to the advantaged group, while condition \eqref{eq:condition4} holds when the decision outcome is negatively correlated with the probability of belonging to the disadvantaged group conditionally on the true protected attribute.
In the example of loan application with race proxy based on geolocation, geolocation is correlated with both race ratio and socioeconomic status (e.g., income, FICO score, etc.): locations that are primarily occupied by the advantaged racial groups tend to be associated with relatively higher socioeconomic status while locations that are primarily occupied by the disadvantaged racial groups tend to be associated with relatively lower socioeconomic status.
As a result, people who live in locations dominant by the advantaged racial group are assigned with high probability of belonging to the advantaged group, and they are more likely to get approved in loan application because of relatively higher socioeconomic status.
Conversely, the loan approval is negatively correlated with the probability of belonging to disadvantaged groups.
This example demonstrates that \eqref{eq:condition3}---\eqref{eq:condition4} are likely to hold when the predictors $Z$ are also strongly correlated with the decision outcome, e.g., the geolocation is correlated with the loan approval because of the socioeconomic status disparities across different geolocations.

\begin{figure}
\centering
\includegraphics[width=\linewidth]{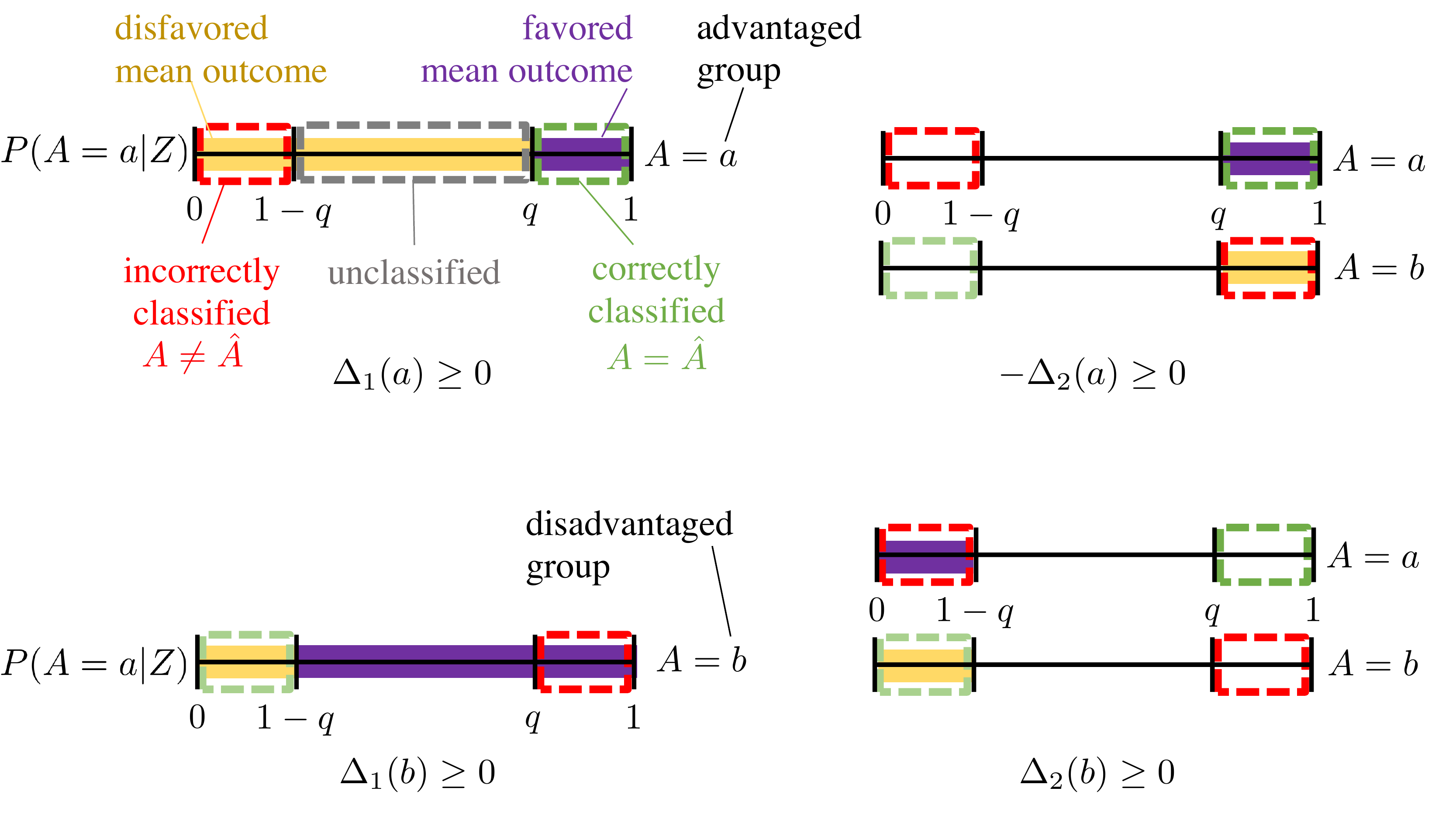}
  \caption{Illustration of conditions \eqref{eq:condition1}---\eqref{eq:condition4} in Corollary \ref{corollary: demographic-est-race}, illustrating the complex interplay between membership in a protected class $A$ and outcome $Y$ that give rise to overestimated demographic disparity using the thresholded estimator of Definition \ref{def:thresholded_estimator}.
  The horizontal axes show $\pr({A}=a|Z)$, the probability that a probabilistic proxy model predicts membership in the advantaged group $a$.
  Dashed boxes represent correct ($\hat{A}=A$, green) and incorrect classifications ($\hat{A}\ne A$, red), from the thresholded estimated membership rule of Definition \ref{def:thresholding_rule}, with the interval $(1-q,q)$ always unclassified.
  Solid bars represent favored (purple) and disfavored mean outcomes (yellow).
} \label{figure: theorem}
\end{figure}

The following corollary shows that \eqref{eq:condition1}---\eqref{eq:condition4} have effects with different magnitudes on the overestimation bias:
\begin{corollary} \label{corollary: C-functions}
Let $A$ be a binary protected class with values $a$ and $b$. For $u \in \{a, b\}$ and $u^c$ as the class opposite to $u$, the quantities $C_1(u), C_2(u), C_3(u)$ in Theorem \ref{thm: demographic-est-race-bias} are related in the following way:
\begin{enumerate}[label=(\roman*)]
    \item $C_2(u) > C_3(u)$ if and only if
        \begin{equation}\label{c_value_23}
        \pr (A = u \mid \hat{A} = u) > \pr (A = u^c \mid \hat{A} = u).
        \end{equation}
    \item $C_2(u) > C_1(u)$ if and only if
        \begin{equation}\label{c_value_21}
        \pr (A = u) > \pr (\hat{A} = u)
        \end{equation}
\end{enumerate}
Thus if the conditions \eqref{c_value_23} and \eqref{c_value_21} both hold, then $C_2(u) > C_1(u)$ and  $C_2(u) >C_3(u)$.
\end{corollary}

Condition \eqref{c_value_23} holds as long as the proxy model is reasonably predictive of the true protected class. Condition \eqref{c_value_21} usually holds for the thresholded estimated membership (Definition \ref{def:thresholding_rule}) when a high threshold $q$ is used: a fairly high fraction of observations are unclassified under the thresholded estimation rule with high threshold $q$, so the fraction of observations  classified into one protected class is usually lower than the fraction of observations actually belonging to that  protected class.
Therefore, when a high threshold $q$ is used, usually conditions \eqref{eq:condition3}---\eqref{eq:condition4} contribute more to the overestimation bias than \eqref{eq:condition1}---\eqref{eq:condition2}.
As a result, even if \eqref{eq:condition1}---\eqref{eq:condition2} are violated, as long as \eqref{eq:condition3}---\eqref{eq:condition4} hold, the thresholded estimator is still very likely to overestimate demographic disparity.
In contrast, when a low threshold is used, the difference between $\pr (A = u) > \pr (\hat{A} = u)$ is usually smaller, and thus $C_2(u) - C_1(u)$ is also smaller.
In this case, \eqref{eq:condition1}---\eqref{eq:condition2} and \eqref{eq:condition3}---\eqref{eq:condition4} may have comparable contribution to the estimation bias.
In Section 4.2, we will verify these facts by using different thresholds on the mortgage dataset.

\section{Numerical results}
\label{sec:numerics}

\subsection{Analysis of bias terms in synthetic data}

In this part, we simulate an example where geolocation is used to construct proxy for race. We use this example to demonstrate different sources of overestimation and underestimation bias of the weighted estimator and the thresholded estimator.

\paragraph{Experimental setup}
We consider race as the unknown protected attribute $A$ with $A = a$ being the advantaged racial group and $A = b$ being the disadvantaged racial group. Suppose there are three different neighborhoods ($Z = z_1, z_2, z_3$) where the proportions of people belonging to the advantaged group are $0.2, 0.5, 0.8$ respectively. These proportions are in turn the race proxy probabilities. For example, all people who live in the neighborhood $z_1$ are assigned with $(\pr (A = a \mid Z = z_1) = 0.2, \pr (A = b \mid Z = z_1)  = 0.8)$ as their race proxy. If we apply the thresholded estimation rule with $q = 0.75$, the estimated races for people living in neighborhoods tracts are:
\begin{center}
\begin{tabular}{|c||c|c|c|}
\hline
 & $Z=z_{1}$ & $Z=z_{2}$ & $Z=z_{3}$\tabularnewline
\hline
\hline
$\hat{A}$ & $b$ & NA & $a$\tabularnewline
\hline
\end{tabular}
\end{center}
We assume that geolocation is strongly correlated with people's income $X$: $\expect (X \mid Z = z_1) < \expect (X \mid Z = z_2) = 2 < \expect (X \mid Z = z_3)$. Therefore, $z_1$ can be considered as the low-income tract where  the proportion of people belonging to the disadvantaged group is high ($\pr (A = b \mid Z = z_3) = 0.8$), $z_2$ can be considered as middle-income tract where the proportions of people belonging to the advantaged group and the disadvantaged group are equal ($\pr (A = a \mid Z = z_2) = \pr (A = b \mid Z = z_2) = 0.5$), and $z_3$ can be considered as high-income tract where the proportion of people belonging to the advantaged group is also high ($\pr (A = a \mid Z = z_1) = 0.8$). The decision outcome $Y$ is loan approval, which we assume solely depends on the income $X$: we simulate $Y$ according to $\pr (Y = 1 \mid X) = 1/(1 + \exp(-\lambda(X - 2)))$ with $\lambda > 0$. Thus  people with higher income are more likely to get approved ($Y = 1$). Here $\lambda$ controls the extent to which loan approval $Y$ depends on income $X$. Since income $X$ solely depends on geolocation $Z$, $\lambda$ indirectly controls the dependence of the outcome $Y$ on geolocation $Z$ and thus on the corresponding race proxy probabilities $(\pr (A = a \mid Z), \pr (A = b \mid Z))$.

In each experiment, we set the total population size of neighborhoods $z_1, z_2, z_3$ to be 3000, 4000, and 5000 respectively. Each experiment is repeated $30$ times and the average estimated demographic disparity from the thresholded estimator, the weighted estimator, and using the true race is shown in Figure \ref{figure: synthetic}.

\begin{experiment}
\label{exp:intra}
The income is normally distributed according to $\mathcal{N}(\mu_X, 0.25)$ with mean value $\mu_X$ depending on both geolocation and race:
\begin{center}
\begin{tabular}{|c||c|c|c|}
\hline
$\mu_{X}$ & $Z=z_{1}$ & $Z=z_{2}$ & $Z=z_{3}$\tabularnewline
\hline
\hline
$A=a$ & $1+d$ & $2+d$ & $3+d$\tabularnewline
\hline
$A=b$ & $1-d$ & $2-d$ & $3-d$\tabularnewline
\hline
\end{tabular}
\end{center}
where $d$ controls the discrepancy of income $X$ between the two race groups within the same geolocation.
We fix $\lambda = 1$ and vary $d$ from $-0.5$ to $0.5$, which indirectly varies the magnitude of intra-geolocation decision outcome variation between the two races.
By construction, $d$ should affect only the $\Delta_1(a), \Delta_1(b)$ in \eqref{eq:condition1}---\eqref{eq:condition2}.
\end{experiment}

\begin{experiment}
\label{exp:inter}
The income is normally distributed according to $\mathcal{N}(\mu_X, 0.25)$ with the mean value $\mu_X$  depending on only geolocation:
\begin{center}
\begin{tabular}{|c||c|c|c|}
\hline
 & $Z=z_{1}$ & $Z=z_{2}$ & $Z=z_{3}$\tabularnewline
\hline
\hline
$\mu_{X}$ & $1$ & $2$ & $3$\tabularnewline
\hline
\end{tabular}
\end{center}
In other words, people with different races within the same geolocation have the same income distribution and thus the same decision outcome distribution. Therefore, there is no \textit{intra-geolocation} decision outcome variation between the two races. As a result, when we vary $\lambda$ from $0.2$ to $2$, only \textit{inter-geolocation} outcome variation changes, which affects $\Delta_2(a), \Delta_2(b)$ in conditions \eqref{eq:condition3}---\eqref{eq:condition4}.
\end{experiment}

\begin{figure}
\centering
\includegraphics[width = \linewidth]{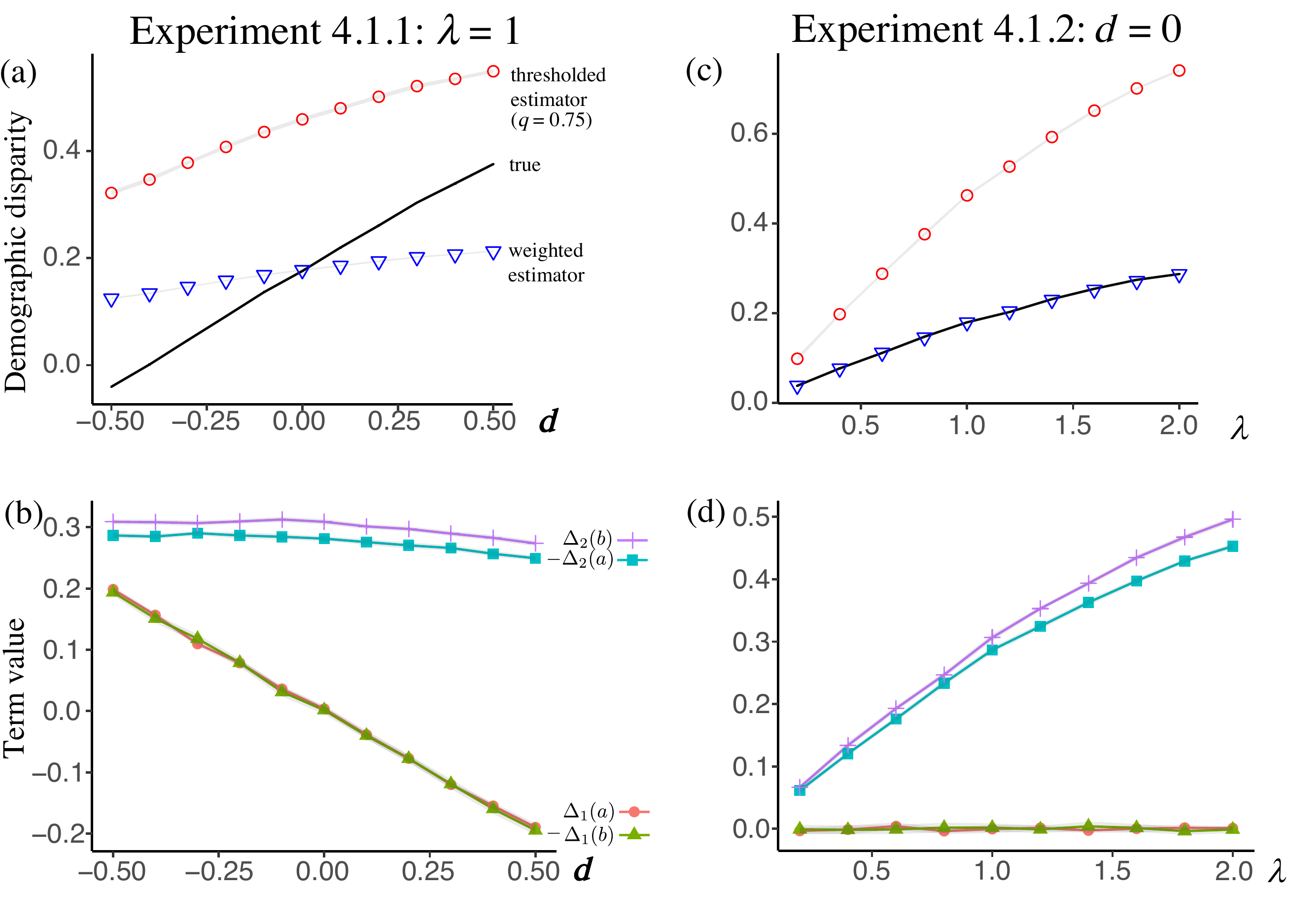}
\caption{Top row: performance of the thresholded ($\hat{\delta}_{q=0.75}$, red circles) and weighted ($\hat{\delta}_W$, blue triangles) estimators of demographic disparity in Experiments \ref{exp:intra} (a) and \ref{exp:inter} (c), relative to the true demographic disparity (black line).
Bottom row: values of the terms $\Delta_1(a)$, $-\Delta_1(b)$, $\Delta_2(a)$, and $-\Delta_2(b)$ in Corollary \ref{corollary: demographic-est-race} in Experiments \ref{exp:intra} (b) and \ref{exp:inter} (d), showing their contributions to the bias.
Shown in light gray is the confidence interval within two standard deviations,
averaged over 30 simulations, which in all cases is comparable to the width of the plotted lines.
}\label{figure: synthetic}
\end{figure}

Figure \ref{figure: synthetic}(b) shows that in Experiment \ref{exp:intra}, only $\Delta_1\!(a)$ and $\Delta_1\!(b)$ vary strongly with $d$, so they are largely responsible for the bias observed in Figure \ref{figure: synthetic}(a).
When $d < 0$, the disadvantaged group have comparatively higher incomes and are thus more likely to be approved in loan application.
This implies that conditions \eqref{eq:condition1}---\eqref{eq:condition2} hold and therefore result in overestimation bias for the thresholded estimator.
Furthermore, $Y$ covaries negatively with $\ind(A=a)$ and positively with $\ind(A=b)$ conditionally on $Z$, also resulting in an overestimation bias for the weighted estimator.
When $d > 0$, however, conditions \eqref{eq:condition1}---\eqref{eq:condition2} are violated.
The terms $\Delta_1(a)$ and $\Delta_1(b)$ start to counteract overestimation bias, thus the overestimation bias decreases with $d$.
Furthermore, $Y$ now covaries positively with $\ind(A=a)$ and negatively with $\ind(A=b)$, also resulting in an underestimation bias for the weighted estimator.

Figure \ref{figure: synthetic}(d) shows that in Experiment \ref{exp:inter}, only the terms $\Delta_2(a)$ and $\Delta_2(b)$ vary strongly with $\lambda$, so these terms are responsible for the variation observed in the thresholded estimator $\hat\delta_q$ in Figure \ref{figure: synthetic}(c).
As $\lambda$ increases, both $-\Delta_2(a)$ and $\Delta_2(b) > 0$ increase, along with the overestimation bias of $\hat\delta_q$.
In contrast, the weighted estimator $\hat\delta_W$ is unbiased because the income distribution does not depend on race, and so $Y$ is independent with $\ind(A = a)$ and $\ind(A=b)$ conditionally on $Z$.

While presented only for the particular choices of $\lambda=1$ for Experiment \ref{exp:intra} and $d=0$ for Experiment \ref{exp:inter}, the observation of which terms vary strongly with $d$ and $\lambda$ hold true for quite a few different choices.

\subsection{Estimation biases in the HMDA mortgage data set with geolocation proxy for race}
\label{sec:mortgage}

In this section, we use the public HMDA (Home Mortgage Disclosure Act) data set\footnote{Data link: \url{https://www.consumerfinance.gov/data-research/hmda/explore}.} to demonstrate demographic disparity estimation bias when using a probabilistic proxy for race.
This data set contains mortgage loan application records in the U.S.\ for which the geolocation (state, county, and census tract), self-reported race/\-ethnicity, and loan origination outcome were reported, and
has been used in the literature to evaluate the BISG race proxy \cite{Baines:2014aa, CFPB:proxy2014, zhang2016assessing}.
We use the loan data for the years 2011---2012, consistent with \cite{CFPB:proxy2014}.
We denote ${Y} = 1$ if a loan application was approved or originated, and ${Y} = 0$ if it was denied. The final sample contains around 17 million observations with non-missing geolocation, race, and loan origination outcome information.

We consider a race proxy based only on geolocation, as the public data set is anonymized and omits surnames.
This proxy is derived from the racial and ethnic composition of the U.S.\ population that is over 18 years of age,
using the census tracts of the 2010 decennial census\footnote{We use the census tract level geolocation-only proxy constructed by the CFPB, as described in \url{https://github.com/cfpb/proxy-methodology}.}.
For example, the proportion of Hispanic, white, black, and API (Asian or Pacific Islander) in census tract 020100, Autauga County, Alabama are 0.02, 0.86, 0.10, 0.01 respectively, according to the 2010 census.
This quadruple is assigned to all applicants from this census tract as their race proxy.
This probabilistic proxy, while different from BISG, is nevertheless sufficient to demonstrate the general result regarding disparity estimation bias in Section \ref{sec:theorems}.
We present results for Hispanic, white, and black subpopulations in this section, and defer the result for API to the Appendix.

\begin{figure}
\centering
\includegraphics[width=\linewidth]{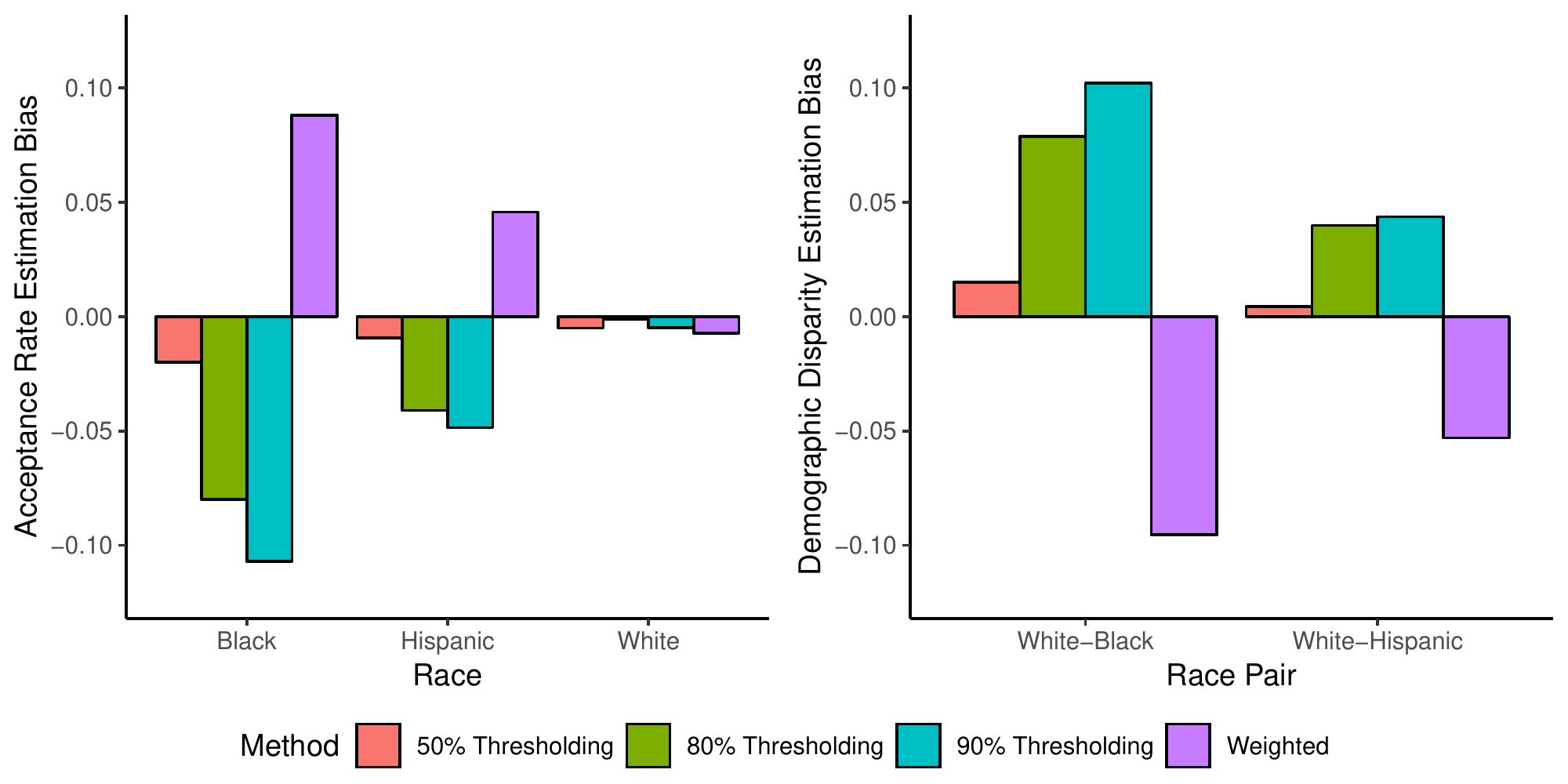}
\caption{Left: estimation biases of loan acceptance rates $\mu$ for different races in the HMDA data set of Section \ref{sec:mortgage}, using the thresholded estimator for mean group outcome $\hat\mu_q$ ($q=0.5, 0.8, 0.9$) from Definition \ref{def:thresholded_estimator}, as well as the weighted estimator $\hat\mu_W$ from Definition \ref{def:weighted_estimator}, relative to the true mean group outcome $\mu$ calculated using the actual race labels.
Right: estimation biases of demographic disparity $\delta$ between pairs of races, using the thresholded estimator $\hat\delta_q$ from Definition \ref{def:thresholded_estimator} and weighted estimator $\hat\delta_W$from Definition \ref{def:weighted_estimator}, relative to the true demographic disparity $\delta$ calculated using the actual race labels.}\label{figure: mortgage_demo_disparity}
\end{figure}

\paragraph{Estimation bias of thresholded estimator and weighted estimator}
In Figure \ref{figure: mortgage_demo_disparity}, we show the estimation bias of the thresholded estimator with different thresholds and the weighted estimator.
Clearly the thresholded estimator underestimates the loan acceptance rate of black and Hispanic groups but it estimates the loan acceptance rate for White group accurately.
As a result, the thresholded estimator overestimates the demographic disparity.
Moreover, the overestimation bias tends to decrease as the threshold $q$ decreases.
In contrast, the weighted estimator displays the opposite performance and underestimates the demographic disparity.

\begin{figure}
\centering
\includegraphics[width=\linewidth]{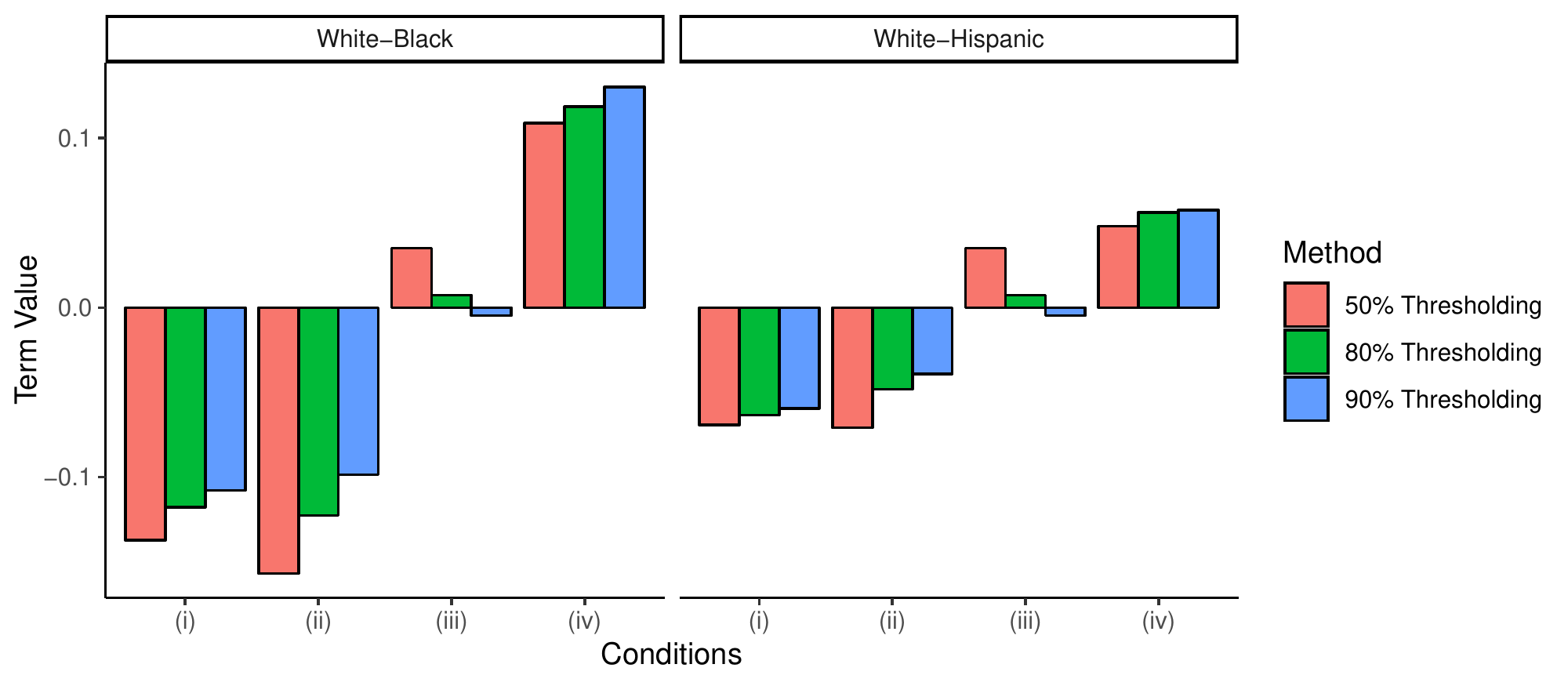}
\caption{The values of the terms $\Delta_1(a)$, $-\Delta_1(b)$, $\Delta_2(a)$, and $-\Delta_2(b)$ in Corollary \ref{corollary: demographic-est-race} for the mortgage data set of Section \ref{sec:mortgage}, for thresholds $q=$ 0.5, 0.8, and 0.9.
Negative values demonstrate partial violation of the conditions of Corollary \ref{corollary: demographic-est-race}.} \label{figure: mortgage_demo_disparity_conditions}
\end{figure}

\begin{figure}
\centering
\includegraphics[width=\linewidth]{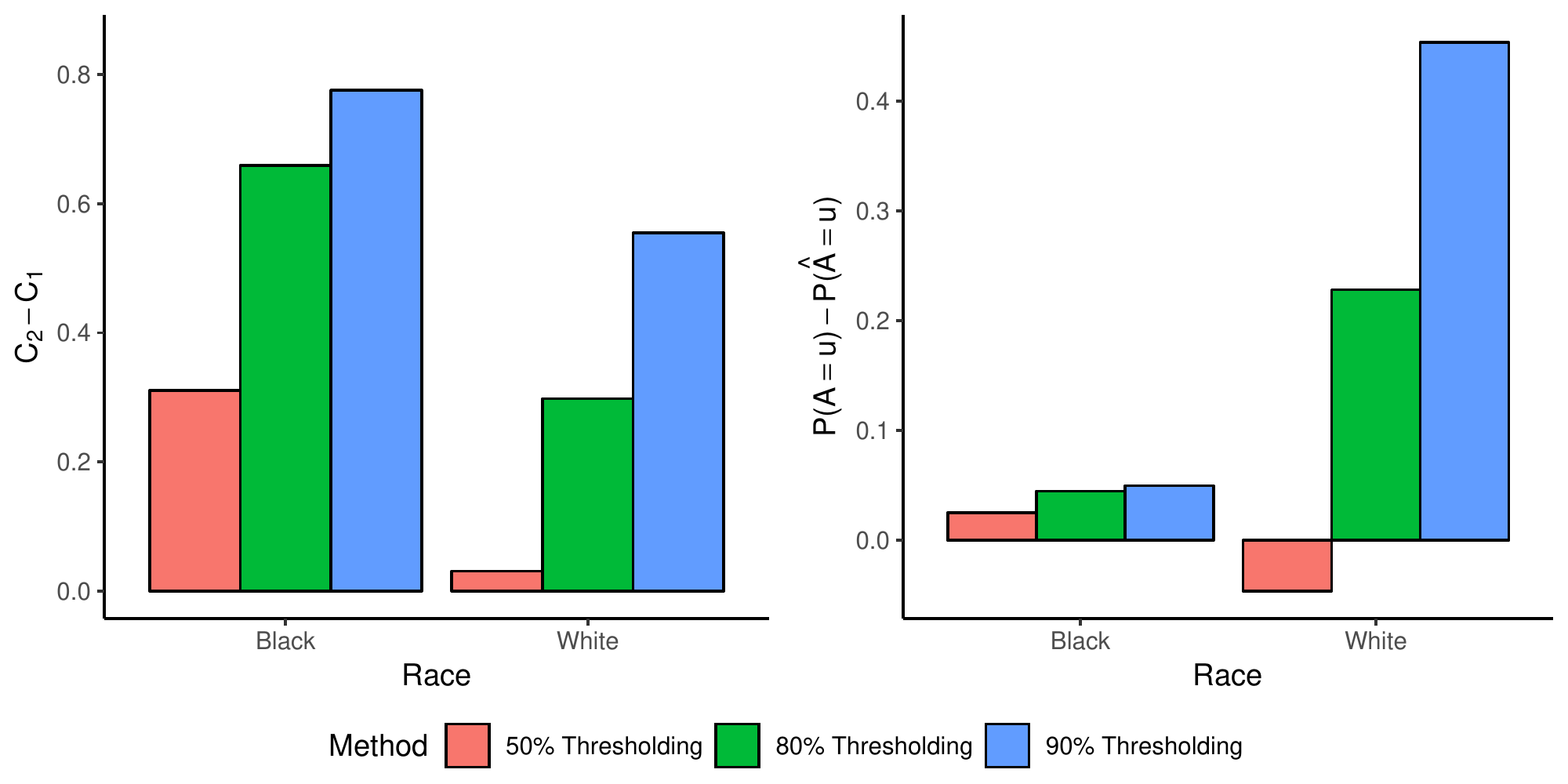}
\caption{Left: the difference between the terms $C_2(u)$ and $C_1(u)$ in Theorem \ref{thm: demographic-est-race-bias} for races $u=$ black or white and thresholds $q=$ 0.5, 0.8 and 0.9 for estimating the race labels according to the thresholded rule of Definition \ref{def:thresholding_rule}.
As $q$ increases, $C_2(u)$ is increasingly larger than $C_1(u)$, and thus the $\Delta_2(u)$ terms contribute more to the estimation bias than the $\Delta_1(u)$ terms.
Right: difference between the probabilities of the true and estimated race labels, $\pr (A = u) - \pr (\hat{A} = u)$. As $q$ increases from 0.5 to 0.9, the proportion of unclassified samples increases from 0.065 to 0.63, which causes $\pr (\hat{A} = u)$ to decrease drastically. As a result, $C_2(u) - C_1(u)$ increases due to Corollary \ref{corollary: C-functions}.}\label{figure: C_values}
\end{figure}

\paragraph{Bias source of thresholded estimator}
Figure \ref{figure: mortgage_demo_disparity_conditions} shows the conditions  \eqref{eq:condition1}---\eqref{eq:condition4}.
Only \eqref{eq:condition4} strictly holds, meaning that people from the disadvantaged group (Hispanic or black) living in census tracts where their race ratio exceeds the corresponding classification threshold have lower average loan acceptance rate than people from the disadvantaged group living in census tracts where their race ratio is not as high as the classification threshold.
This captures the \textit{inter-geolocation} outcome variation: the loan approval is negatively correlated with the disadvantaged group prevalence across the geolocations.
In Example \ref{ex:income}, we give a reason why loan acceptance is correlated with the race proxy probabilities: geolocation is correlated with both race proportions (i.e., race proxy probabilities) and socioeconomic status (e.g., income, FICO score, etc.) that affects loan approval. In Appendix \ref{section:correlation_income}, we validate that such correlations are indeed obvious in the mortgage data set.

In contrast to conditions \eqref{eq:condition4}, conditions \eqref{eq:condition1}---\eqref{eq:condition2} are strictly violated, and condition \eqref{eq:condition3} is slightly violated.
However, the thresholded estimators still have overestimation bias because \eqref{eq:condition3}---\eqref{eq:condition4} have dominant effects, as shown in Figure \ref{figure: C_values}.
Moreover, as the threshold $q$ increases, conditions \eqref{eq:condition3}---\eqref{eq:condition4} start to dominate, increasing the overestimation bias.
Nevertheless, the apparent reduction of overestimation bias when using lower threshold does not mean that using lower threshold is better in practice. Instead, the reduction of overestimation bias is the consequence of delicate counterbalance between the two opposing bias sources: the violation of conditions \eqref{eq:condition1}---\eqref{eq:condition2} contributes to underestimation bias and the conditions \eqref{eq:condition3}---\eqref{eq:condition4} contribute to overestimation bias. Thus, the smaller bias from using a lower threshold is not a robust finding. For example, for the experiments in Section \ref{sec:numerics}, changing the threshold from $0.75$ to $0.5$ does not affect either the race imputation result or the demographic disparity estimation bias. This reflects the intrinsic complexity of the thresholded estimator.

In summary, according to Corollary \ref{corollary: C-functions}, the fact that the thresholded estimation rule for race excludes many unclassified examples makes the overestimation predominantly be determined by the inter-geolocation outcome variation captured by conditions \eqref{eq:condition3}---\eqref{eq:condition4}, especially when the threshold $q$ is very high. Moreover, strong racial segregation and socioeconomic status disparities across different geolocations make condition \eqref{eq:condition3} or \eqref{eq:condition4} (if not both) very likely to hold. As a result, the thresholded estimator tends to overestimate the demographic disparity, especially when high threshold is used. Since BISG also uses geolocation for race estimation, the reported overestimation bias of using thresholded estimator with BISG should be at least largely (if not all) due to the same reason. Moreover, the estimation bias of thresholded can be very sensitive to the threshold value because the threshold value influences the interplay of the bias sources in \eqref{eq:condition1}---\eqref{eq:condition4}. This reflects the intrinsic limitation of the thresholded estimator.


\begin{figure}
  \centering
    \includegraphics[width=\linewidth]{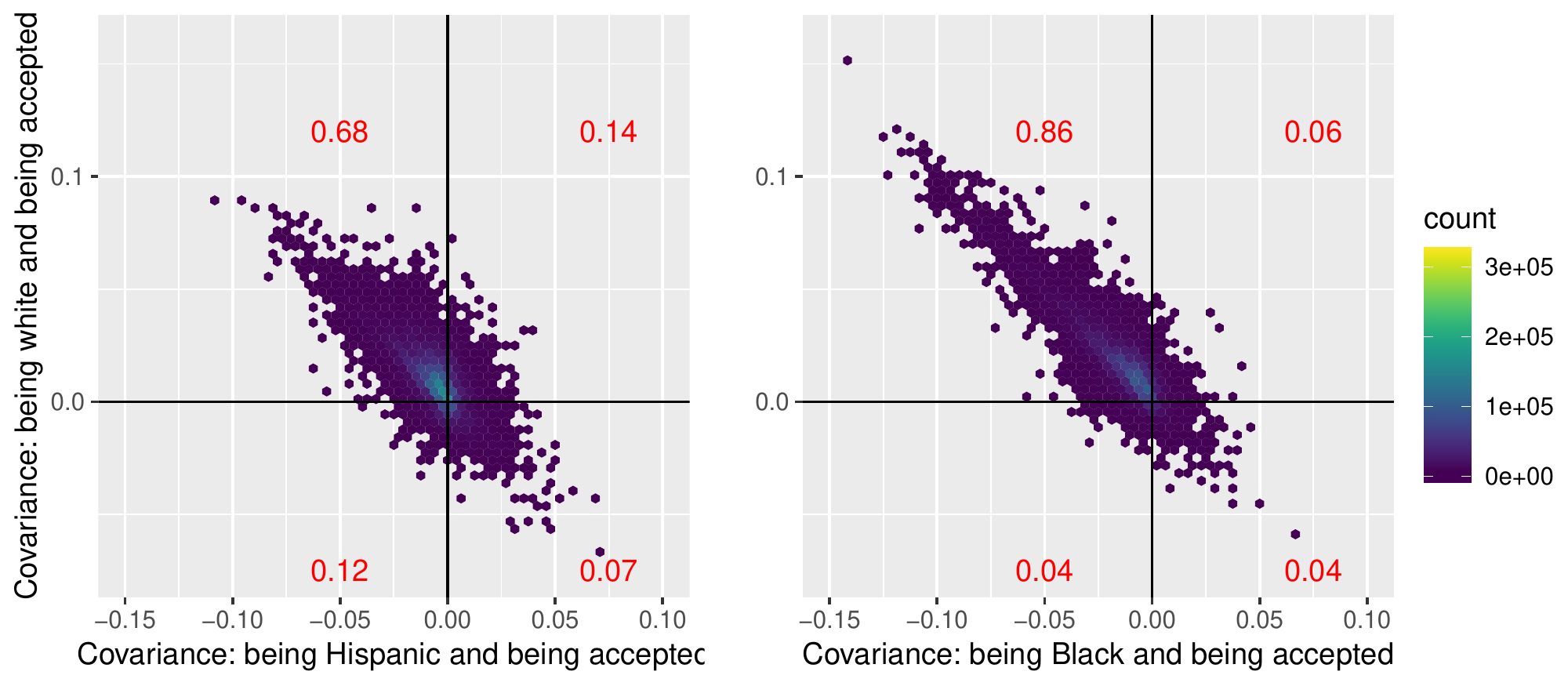}
  \caption{The population distribution across different combinations of within-census-tract-covariances between the loan acceptance and the membership to advantaged group or the membership to disadvantaged group. The red numbers in each quadrant represents the proportion of people who live in census tracts with the corresponding covariance combination.} \label{figure: mortgage_weighted_conditions}
\end{figure}

\paragraph{Bias source of the weighted estimator}
Figure \ref{figure: mortgage_weighted_conditions} shows the distribution of population who live in census tracts with different combinations of $\cov (\ind(A = a), {Y} \mid Z)$ and $\cov (\ind(A = b), {Y} \mid Z)$ with $a$ representing white and $b$ representing black or Hispanic.
Theorem \ref{thm: weighted} implies that the upper left quadrants account for the underestimation bias of weighted estimator for demographic disparity,
as these quadrants contain all census tracts where the loan acceptance is positively correlated with membership to the advantaged group ($\cov (\ind(A = a), {Y} \mid Z) > 0$),
while negatively correlated with the membership to the disadvantaged group ($\cov (\ind(A = b), {Y} \mid Z) < 0$).
Figure \ref{figure: mortgage_weighted_conditions} shows that the majority of people live in census tracts falling in the upper left quadrants.
In other words, most people live in census tracts where white people have a higher average loan approval rate than their black or Hispanic neighbors. This explains the underestimation bias of the weighted estimator since the estimation bias of the weighted estimator is solely determined by the intra-geolocation outcome variation according to Theorem \ref{thm: weighted}.


\paragraph{Summary}
The empirical results show that our theoretical analysis provides convincing explanations for the observed estimation bias of the thresholded estimator and the weighted estimator. We show that the bias sources of the weighted estimator and the thresholded estimator are very different: the bias of the weighted estimator is solely determined by the intra-geolocation variation, and the bias of the thresholded estimator is determined by both inter-geolocation variation and intra-geolocation variation. When high threshold is used, the inter-geolocation variation dominates.
In the mortgage dataset, we show that the racial segregation, socioeconomic status, outcome disparity pattern with respect to geolocation make the thresholded estimator tend to overestimate the demographic disparity, and the weighted estimator tend to underestimate the demographic disparity.
This explains the overestimation bias of the thresholded estimator with BISG reported by previous literature.
Moreover, our results show that the estimation bias of thresholded estimator is sensitive to the threshold because of complex interplay of different bias sources.
As a result, the observed estimation bias pattern in one setting may hardly generalize to other settings.
In contrast, the weighted estimator only has one bias source, and is thus easier to reason about. 




\section{Conclusions}

This paper presents the first theoretical analysis of bias in outcome disparity assessments using a probabilistic proxy model for the unobserved protected class.
In Theorem \ref{thm: demographic-est-race-bias}, we derived the bias of a thresholded estimator (Definition \ref{def:thresholded_estimator}) that has been described in the literature.
We also gave sufficient conditions in Corollary \ref{corollary: demographic-est-race} to understand when this methodology is biased,
and to what extent.
Our theoretical analysis is valid whenever a proxy model is used with a thresholded estimator to impute protected class membership,
and is thus consistent with previous studies that had observed an overestimation bias of the thresholded estimator based on BISG.
When applied to the public HMDA mortgage dataset with geolocation as the sole proxy for race membership (Section \ref{sec:mortgage}),
we found that the estimation bias of the thresholded estimator depends on the complex interaction of multiple different biases,
producing a strong sensitivity to the precise value of the threshold used.
Further studies will be needed to demonstrate the robustness of this numerical finding,
particularly when using other proxy models and other measures of fairness.
Nevertheless, our work signals caution for choosing the \textit{ad hoc} thresholds which are used in practice,
as without the ground truth labels, we are unable to determine the optimal choice of threshold. 

To alleviate the theoretical challenges of the thresholded estimator, we also proposed a weighted estimator (Definition 2.5),
which propagated the uncertainty resulting from the probabilistic proxy onto the final estimand.
We found that the estimation bias of this weighted estimator has only one bias source,
arising from the loan approval discrepancy between difference races within the same geolocation,
which led to an overall underestimate.
As the behavior of this estimator's bias is simpler to reason about, 
we believe that the weighted estimator may be a useful new method to incorporate into outcome disparity evaluations when proxy models are used,
especially if the sign of estimation bias can be determined with external knowledge.

\bibliographystyle{ACM-Reference-Format}
\bibliography{bib}


\begin{thebibliography}{32}


\ifx \showCODEN    \undefined \def \showCODEN     #1{\unskip}     \fi
\ifx \showDOI      \undefined \def \showDOI       #1{#1}\fi
\ifx \showISBNx    \undefined \def \showISBNx     #1{\unskip}     \fi
\ifx \showISBNxiii \undefined \def \showISBNxiii  #1{\unskip}     \fi
\ifx \showISSN     \undefined \def \showISSN      #1{\unskip}     \fi
\ifx \showLCCN     \undefined \def \showLCCN      #1{\unskip}     \fi
\ifx \shownote     \undefined \def \shownote      #1{#1}          \fi
\ifx \showarticletitle \undefined \def \showarticletitle #1{#1}   \fi
\ifx \showURL      \undefined \def \showURL       {\relax}        \fi
\providecommand\bibfield[2]{#2}
\providecommand\bibinfo[2]{#2}
\providecommand\natexlab[1]{#1}
\providecommand\showeprint[2][]{arXiv:#2}

\bibitem[\protect\citeauthoryear{Baines and Courchane}{Baines and Courchane}{4
  11}]%
        {Baines:2014aa}
\bibfield{author}{\bibinfo{person}{Arthur~P Baines} {and}
  \bibinfo{person}{Marsha~J Courchane}.} \bibinfo{year}{2014-11}\natexlab{}.
\newblock \bibinfo{title}{Fair Lending: Implications for the Indirect Auto
  Finance Market}.
\newblock
\newblock
\urldef\tempurl%
\url{https://www.crai.com/sites/default/files/publications/Fair-Lending-Implications-for-the-Indirect-Auto-Finance-Market.pdf}
\showURL{%
\tempurl}


\bibitem[\protect\citeauthoryear{Barocas and Selbst}{Barocas and
  Selbst}{2016}]%
        {Barocas:2016aa}
\bibfield{author}{\bibinfo{person}{Solon Barocas} {and} \bibinfo{person}{Andrew
  Selbst}.} \bibinfo{year}{2016}\natexlab{}.
\newblock \showarticletitle{Big Data's Disparate Impact}.
\newblock \bibinfo{journal}{\emph{California Law Review}}
  \bibinfo{volume}{104}, \bibinfo{number}{1} (\bibinfo{year}{2016}),
  \bibinfo{pages}{671--729}.
\newblock
\urldef\tempurl%
\url{https://doi.org/10.15779/Z38BG31}
\showDOI{\tempurl}


\bibitem[\protect\citeauthoryear{Berk, Heidari, Jabbari, Kearns, and Roth}{Berk
  et~al\mbox{.}}{2018}]%
        {Berk2017}
\bibfield{author}{\bibinfo{person}{Richard Berk}, \bibinfo{person}{Hoda
  Heidari}, \bibinfo{person}{Shahin Jabbari}, \bibinfo{person}{Michael Kearns},
  {and} \bibinfo{person}{Aaron Roth}.} \bibinfo{year}{2018}\natexlab{}.
\newblock \showarticletitle{Fairness in Criminal Justice Risk Assessments: The
  State of the Art}.
\newblock \bibinfo{journal}{\emph{Sociological Methods \& Research}}
  (\bibinfo{year}{2018}), \bibinfo{pages}{1--42}.
\newblock
\urldef\tempurl%
\url{https://doi.org/10.1177/0049124118782533}
\showDOI{\tempurl}
\showeprint[arxiv]{1703.09207}


\bibitem[\protect\citeauthoryear{Bickel, Hammel, and O'Connell}{Bickel
  et~al\mbox{.}}{1975}]%
        {Bickel1975}
\bibfield{author}{\bibinfo{person}{P.~J. Bickel}, \bibinfo{person}{E.~A.
  Hammel}, {and} \bibinfo{person}{J.~W. O'Connell}.}
  \bibinfo{year}{1975}\natexlab{}.
\newblock \showarticletitle{Sex Bias in Graduate Admissions: Data from
  {B}erkeley}.
\newblock \bibinfo{journal}{\emph{Science}} \bibinfo{volume}{187},
  \bibinfo{number}{4175} (\bibinfo{year}{1975}), \bibinfo{pages}{398--404}.
\newblock
\urldef\tempurl%
\url{https://doi.org/10.1126/science.187.4175.398}
\showDOI{\tempurl}


\bibitem[\protect\citeauthoryear{Calders and Verwer}{Calders and
  Verwer}{2010}]%
        {Calders2010}
\bibfield{author}{\bibinfo{person}{Toon Calders} {and} \bibinfo{person}{Sicco
  Verwer}.} \bibinfo{year}{2010}\natexlab{}.
\newblock \showarticletitle{Three naive {B}ayes approaches for
  discrimination-free classification}.
\newblock \bibinfo{journal}{\emph{Data Mining and Knowledge Discovery}}
  \bibinfo{volume}{21}, \bibinfo{number}{2} (\bibinfo{year}{2010}),
  \bibinfo{pages}{277--292}.
\newblock
\showISSN{1384-5810}
\urldef\tempurl%
\url{https://doi.org/10.1007/s10618-010-0190-x}
\showDOI{\tempurl}


\bibitem[\protect\citeauthoryear{Chen}{Chen}{2018}]%
        {Chen2018}
\bibfield{author}{\bibinfo{person}{Jiahao Chen}.}
  \bibinfo{year}{2018}\natexlab{}.
\newblock \showarticletitle{Fair Lending Needs Explainable Models for
  Responsible Recommendation}. In \bibinfo{booktitle}{\emph{Proceedings of the
  2nd FATREC Workshop on Responsible Recommendation}}
  \emph{(\bibinfo{series}{FATREC'18})}. \bibinfo{publisher}{ACM}.
\newblock
\showeprint[arxiv]{1808.04684}


\bibitem[\protect\citeauthoryear{Chouldechova}{Chouldechova}{2017}]%
        {Chouldechova2017}
\bibfield{author}{\bibinfo{person}{Alexandra Chouldechova}.}
  \bibinfo{year}{2017}\natexlab{}.
\newblock \showarticletitle{Fair prediction with disparate impact: A study of
  bias in recidivism prediction instruments}.
\newblock \bibinfo{journal}{\emph{Big Data}} \bibinfo{volume}{5},
  \bibinfo{number}{2} (\bibinfo{year}{2017}).
\newblock
\urldef\tempurl%
\url{https://doi.org/10.1089/big.2016.0047}
\showDOI{\tempurl}
\showeprint[arxiv]{1703.00056}


\bibitem[\protect\citeauthoryear{{Committee on Financial Services}}{{Committee
  on Financial Services}}{5 11}]%
        {CFS:2015}
\bibfield{author}{\bibinfo{person}{{Committee on Financial Services}}.}
  \bibinfo{year}{2015-11}\natexlab{}.
\newblock \bibinfo{title}{Unsafe at Any Bureaucracy: {CFPB} Junk Science and
  Indirect Auto Lending}.
\newblock
\newblock
\urldef\tempurl%
\url{https://financialservices.house.gov/uploadedfiles/11-24-15_cfpb_indirect_auto_staff_report.pdf}
\showURL{%
\tempurl}


\bibitem[\protect\citeauthoryear{{Committee on Financial Services}}{{Committee
  on Financial Services}}{6 01}]%
        {CFS:2016}
\bibfield{author}{\bibinfo{person}{{Committee on Financial Services}}.}
  \bibinfo{year}{2016-01}\natexlab{}.
\newblock \bibinfo{title}{Unsafe at Any Bureaucracy, {Part II}: How the Bureau
  of Consumer Financial Protection Removed Anti-fraud Safeguards to Achieve
  Political Goals}.
\newblock
\newblock
\urldef\tempurl%
\url{https://financialservices.house.gov/uploadedfiles/cfpb_indirect_auto_part_ii.pdf}
\showURL{%
\tempurl}


\bibitem[\protect\citeauthoryear{{Committee on Financial Services}}{{Committee
  on Financial Services}}{7 01}]%
        {CFS:2017}
\bibfield{author}{\bibinfo{person}{{Committee on Financial Services}}.}
  \bibinfo{year}{2017-01}\natexlab{}.
\newblock \bibinfo{title}{Unsafe at Any Bureaucracy, Part {III}: The {CFPB}'s
  Vitiated Legal Case Against Auto-Lenders}.
\newblock
\newblock
\urldef\tempurl%
\url{https://financialservices.house.gov/uploadedfiles/1-18-17_cfpb_indirect_auto_staff_report_iii.pdf}
\showURL{%
\tempurl}


\bibitem[\protect\citeauthoryear{{Consumer Financial Protection
  Bureau}}{{Consumer Financial Protection Bureau}}{3 12}]%
        {CFPB:ally2013}
\bibfield{author}{\bibinfo{person}{{Consumer Financial Protection Bureau}}.}
  \bibinfo{year}{2013-12}\natexlab{}.
\newblock \bibinfo{title}{{CFPB} and {DOJ} Order Ally to Pay \$80 Million to
  Consumers Harmed by Discriminatory Auto Loan Pricing}.
\newblock
\newblock
\urldef\tempurl%
\url{https://www.consumerfinance.gov/about-us/newsroom/cfpb-and-doj-order-ally-to-pay-80-million-to-consumers-harmed-by-discriminatory-auto-loan-pricing/}
\showURL{%
\tempurl}


\bibitem[\protect\citeauthoryear{{Consumer Financial Protection
  Bureau}}{{Consumer Financial Protection Bureau}}{2014}]%
        {CFPB:proxy2014}
\bibfield{author}{\bibinfo{person}{{Consumer Financial Protection Bureau}}.}
  \bibinfo{year}{2014}\natexlab{}.
\newblock \bibinfo{title}{Using publicly available information to proxy for
  unidentified race and ethnicity: a methodology and assessment}.
\newblock
\newblock
\urldef\tempurl%
\url{https://www.consumerfinance.gov/data-research/research-reports/using-publicly-available-information-to-proxy-for-unidentified-race-and-ethnicity/}
\showURL{%
\tempurl}


\bibitem[\protect\citeauthoryear{{Consumer Financial Protection
  Bureau}}{{Consumer Financial Protection Bureau}}{8 05}]%
        {CFPB:2018-05}
\bibfield{author}{\bibinfo{person}{{Consumer Financial Protection Bureau}}.}
  \bibinfo{year}{2018-05}\natexlab{}.
\newblock \bibinfo{title}{Statement of the Bureau of Consumer Financial
  Protection on enactment of {S.J.} Res. 57}.
\newblock
\newblock
\urldef\tempurl%
\url{https://www.consumerfinance.gov/about-us/newsroom/statement-bureau-consumer-financial-protection-enactment-sj-res-57/}
\showURL{%
\tempurl}


\bibitem[\protect\citeauthoryear{Conway and Roberts}{Conway and
  Roberts}{1983}]%
        {Conway1983}
\bibfield{author}{\bibinfo{person}{Delores~A Conway} {and}
  \bibinfo{person}{Harry~V Roberts}.} \bibinfo{year}{1983}\natexlab{}.
\newblock \showarticletitle{Reverse Regression, Fairness, and Employment
  Discrimination}.
\newblock \bibinfo{journal}{\emph{Journal of Business {\&} Economic
  Statistics}} \bibinfo{volume}{1}, \bibinfo{number}{1} (\bibinfo{year}{1983}),
  \bibinfo{pages}{75--85}.
\newblock
\urldef\tempurl%
\url{https://doi.org/10.1080/07350015.1983.10509326}
\showDOI{\tempurl}


\bibitem[\protect\citeauthoryear{{Division of Consumer and Community
  Affairs}}{{Division of Consumer and Community Affairs}}{1 07}]%
        {regulationb}
\bibfield{author}{\bibinfo{person}{{Division of Consumer and Community
  Affairs}}.} \bibinfo{year}{2011-07}\natexlab{}.
\newblock \bibinfo{title}{12 {CFR} Supplement {I} to Part 202 - Official Staff
  Interpretations}.
\newblock
\newblock
\urldef\tempurl%
\url{https://www.law.cornell.edu/cfr/text/12/appendix-Supplement_I_to_part_202}
\showURL{%
\tempurl}


\bibitem[\protect\citeauthoryear{Dressel and Farid}{Dressel and Farid}{2018}]%
        {Dressel2018}
\bibfield{author}{\bibinfo{person}{Julia Dressel} {and} \bibinfo{person}{Hany
  Farid}.} \bibinfo{year}{2018}\natexlab{}.
\newblock \showarticletitle{The accuracy, fairness, and limits of predicting
  recidivism}.
\newblock \bibinfo{journal}{\emph{Science Advances}} \bibinfo{volume}{4},
  \bibinfo{number}{1} (\bibinfo{year}{2018}), \bibinfo{pages}{eaao5580}.
\newblock
\urldef\tempurl%
\url{https://doi.org/10.1126/sciadv.aao5580}
\showDOI{\tempurl}


\bibitem[\protect\citeauthoryear{Dwork, Hardt, Pitassi, Reingold, and
  Zemel}{Dwork et~al\mbox{.}}{2012}]%
        {dwork2012fairness}
\bibfield{author}{\bibinfo{person}{Cynthia Dwork}, \bibinfo{person}{Moritz
  Hardt}, \bibinfo{person}{Toniann Pitassi}, \bibinfo{person}{Omer Reingold},
  {and} \bibinfo{person}{Richard Zemel}.} \bibinfo{year}{2012}\natexlab{}.
\newblock \showarticletitle{Fairness through awareness}. In
  \bibinfo{booktitle}{\emph{Proceedings of the 3rd Innovations in Theoretical
  Computer Science Conference}}. ACM, \bibinfo{pages}{214--226}.
\newblock
\urldef\tempurl%
\url{https://doi.org/10.1145/2090236.2090255}
\showDOI{\tempurl}


\bibitem[\protect\citeauthoryear{Elliott, Fremont, Morrison, Pantoja, and
  Lurie}{Elliott et~al\mbox{.}}{2008}]%
        {elliott2008new}
\bibfield{author}{\bibinfo{person}{Marc~N. Elliott}, \bibinfo{person}{Allen
  Fremont}, \bibinfo{person}{Peter~A Morrison}, \bibinfo{person}{Philip
  Pantoja}, {and} \bibinfo{person}{Nicole Lurie}.}
  \bibinfo{year}{2008}\natexlab{}.
\newblock \showarticletitle{A new method for estimating race/ethnicity and
  associated disparities where administrative records lack self-reported
  race/ethnicity}.
\newblock \bibinfo{journal}{\emph{Health Services Research}}
  \bibinfo{volume}{43}, \bibinfo{number}{5p1} (\bibinfo{year}{2008}),
  \bibinfo{pages}{1722--1736}.
\newblock
\urldef\tempurl%
\url{https://doi.org/10.1111/j.1475-6773.2008.00854.x}
\showDOI{\tempurl}


\bibitem[\protect\citeauthoryear{Elliott, Morrison, Fremont, McCaffrey,
  Pantoja, and Lurie}{Elliott et~al\mbox{.}}{9 04}]%
        {Elliott:2009aa}
\bibfield{author}{\bibinfo{person}{Marc~N. Elliott}, \bibinfo{person}{Peter~A.
  Morrison}, \bibinfo{person}{Allen Fremont}, \bibinfo{person}{Daniel~F.
  McCaffrey}, \bibinfo{person}{Philip Pantoja}, {and} \bibinfo{person}{Nicole
  Lurie}.} \bibinfo{year}{2009-04}\natexlab{}.
\newblock \showarticletitle{Using the Census Bureau's surname list to improve
  estimates of race/ethnicity and associated disparities}.
\newblock \bibinfo{journal}{\emph{Health Services and Outcomes Research
  Methodology}} \bibinfo{volume}{9}, \bibinfo{number}{2}
  (\bibinfo{year}{2009-04}), \bibinfo{pages}{69--83}.
\newblock
\urldef\tempurl%
\url{https://doi.org/10.1007/s10742-009-0047-1}
\showDOI{\tempurl}


\bibitem[\protect\citeauthoryear{Greene}{Greene}{1984}]%
        {Greene1984}
\bibfield{author}{\bibinfo{person}{William~H. Greene}.}
  \bibinfo{year}{1984}\natexlab{}.
\newblock \showarticletitle{Reverse regression: The algebra of discrimination}.
\newblock \bibinfo{journal}{\emph{Journal of Business and Economic Statistics}}
  \bibinfo{volume}{2}, \bibinfo{number}{2} (\bibinfo{year}{1984}),
  \bibinfo{pages}{117--120}.
\newblock
\urldef\tempurl%
\url{https://doi.org/10.1080/07350015.1984.10509378}
\showDOI{\tempurl}


\bibitem[\protect\citeauthoryear{Hardt, Price, and Srebro}{Hardt
  et~al\mbox{.}}{2016}]%
        {hardt2016equality}
\bibfield{author}{\bibinfo{person}{Moritz Hardt}, \bibinfo{person}{Eric Price},
  {and} \bibinfo{person}{Nathan Srebro}.} \bibinfo{year}{2016}\natexlab{}.
\newblock \showarticletitle{Equality of opportunity in supervised learning}. In
  \bibinfo{booktitle}{\emph{Advances in Neural Information Processing
  Systems}}. \bibinfo{pages}{3315--3323}.
\newblock
\showeprint[arxiv]{1610.02413}


\bibitem[\protect\citeauthoryear{Kamishima, Akaho, Asoh, and Sakuma}{Kamishima
  et~al\mbox{.}}{2012}]%
        {Kamishima2012}
\bibfield{author}{\bibinfo{person}{Toshihiro Kamishima},
  \bibinfo{person}{Shotaro Akaho}, \bibinfo{person}{Hideki Asoh}, {and}
  \bibinfo{person}{Jun Sakuma}.} \bibinfo{year}{2012}\natexlab{}.
\newblock \showarticletitle{Fairness-Aware Classifier with Prejudice Remover
  Regularizer}. In \bibinfo{booktitle}{\emph{Machine Learning and Knowledge
  Discovery in Databases}}, \bibfield{editor}{\bibinfo{person}{Peter~A Flach},
  \bibinfo{person}{Tijl {De Bie}}, {and} \bibinfo{person}{Nello Cristianini}}
  (Eds.). \bibinfo{publisher}{Springer Berlin Heidelberg},
  \bibinfo{address}{Berlin, Heidelberg}, \bibinfo{pages}{35--50}.
\newblock
\urldef\tempurl%
\url{https://doi.org/10.1007/978-3-642-33486-3_3}
\showDOI{\tempurl}


\bibitem[\protect\citeauthoryear{Kilbertus, Gasc{\'o}n, Kusner, Veale, Gummadi,
  and Weller}{Kilbertus et~al\mbox{.}}{2018}]%
        {kilbertus2018blind}
\bibfield{author}{\bibinfo{person}{Niki Kilbertus}, \bibinfo{person}{Adri{\`a}
  Gasc{\'o}n}, \bibinfo{person}{Matt~J Kusner}, \bibinfo{person}{Michael
  Veale}, \bibinfo{person}{Krishna~P Gummadi}, {and} \bibinfo{person}{Adrian
  Weller}.} \bibinfo{year}{2018}\natexlab{}.
\newblock \showarticletitle{Blind Justice: Fairness with Encrypted Sensitive
  Attributes}, In \bibinfo{booktitle}{Proceedings of the 35th International
  Conference on Machine Learning}.
\newblock \bibinfo{journal}{\emph{Proceedings of Machine Learning Research}}
  \bibinfo{volume}{80}, \bibinfo{pages}{2630--2639}.
\newblock
\showeprint[arxiv]{1806.03281}


\bibitem[\protect\citeauthoryear{Koren}{Koren}{6 08}]%
        {Koren:2016aa}
\bibfield{author}{\bibinfo{person}{James~Rufus Koren}.}
  \bibinfo{year}{2016-08}\natexlab{}.
\newblock \showarticletitle{Feds use {R}and formula to spot discrimination. The
  {GOP} calls it junk science}.
\newblock \bibinfo{journal}{\emph{Los Angeles Times}}
  (\bibinfo{year}{2016-08}).
\newblock
\urldef\tempurl%
\url{http://www.latimes.com/business/la-fi-rand-elliott-20160824-snap-story.html}
\showURL{%
\tempurl}


\bibitem[\protect\citeauthoryear{Lipton, Chouldechova, and McAuley}{Lipton
  et~al\mbox{.}}{2017}]%
        {lipton2018does}
\bibfield{author}{\bibinfo{person}{Zachary~C Lipton},
  \bibinfo{person}{Alexandra Chouldechova}, {and} \bibinfo{person}{Julian
  McAuley}.} \bibinfo{year}{2017}\natexlab{}.
\newblock \showarticletitle{Does mitigating {ML}'s impact disparity require
  treatment disparity?}
\newblock  (\bibinfo{year}{2017}).
\newblock
\showeprint[arxiv]{1711.07076}


\bibitem[\protect\citeauthoryear{Munoz, Smith, and Patil}{Munoz
  et~al\mbox{.}}{2016}]%
        {Munoz2016}
\bibfield{author}{\bibinfo{person}{Cecilia Munoz}, \bibinfo{person}{Megan
  Smith}, {and} \bibinfo{person}{DJ Patil}.} \bibinfo{year}{2016}\natexlab{}.
\newblock \bibinfo{booktitle}{\emph{Big Data: A Report on Algorithmic Systems,
  Opportunity, and Civil Rights}}.
\newblock \bibinfo{type}{{T}echnical {R}eport} May.
\newblock
\urldef\tempurl%
\url{https://www.whitehouse.gov/sites/default/files/microsites/ostp/2016_0504_data_discrimination.pdf}
\showURL{%
\tempurl}


\bibitem[\protect\citeauthoryear{Simpson}{Simpson}{1951}]%
        {Simpson1951}
\bibfield{author}{\bibinfo{person}{Edward~Hugh Simpson}.}
  \bibinfo{year}{1951}\natexlab{}.
\newblock \showarticletitle{The Interpretation of Interaction in Contingency
  Tables}.
\newblock \bibinfo{journal}{\emph{Journal of the Royal Statistical Society
  Series B}} \bibinfo{volume}{13}, \bibinfo{number}{2} (\bibinfo{year}{1951}),
  \bibinfo{pages}{238--241}.
\newblock
\urldef\tempurl%
\url{http://www.jstor.org/stable/2984065}
\showURL{%
\tempurl}


\bibitem[\protect\citeauthoryear{{US Congress}}{{US Congress}}{1968}]%
        {fha}
\bibfield{author}{\bibinfo{person}{{US Congress}}.}
  \bibinfo{year}{1968}\natexlab{}.
\newblock \bibinfo{title}{{42 U.S.C. §3601 ff.: Fair Housing Act}}.
\newblock
\newblock
\urldef\tempurl%
\url{https://www.justice.gov/crt/fair-housing-act-2}
\showURL{%
\tempurl}


\bibitem[\protect\citeauthoryear{{US Congress}}{{US Congress}}{4 10}]%
        {ecoa}
\bibfield{author}{\bibinfo{person}{{US Congress}}.}
  \bibinfo{year}{1974-10}\natexlab{}.
\newblock \bibinfo{title}{{15 U.S.C. §1691 ff.: Equal Credit Opportunity
  Act}}.
\newblock
\newblock
\urldef\tempurl%
\url{https://www.ecfr.gov/cgi-bin/text-idx?tpl=/ecfrbrowse/Title12/12cfr202_main_02.tpl}
\showURL{%
\tempurl}


\bibitem[\protect\citeauthoryear{Zafar, Valera, Rodriguez, and Gummadi}{Zafar
  et~al\mbox{.}}{2017}]%
        {zafar2016learning}
\bibfield{author}{\bibinfo{person}{Muhammad~Bilal Zafar},
  \bibinfo{person}{Isabel Valera}, \bibinfo{person}{Manuel~Gomez Rodriguez},
  {and} \bibinfo{person}{Krishna~P Gummadi}.} \bibinfo{year}{2017}\natexlab{}.
\newblock \showarticletitle{Fairness Constraints: Mechanisms for Fair
  Classification}, In \bibinfo{booktitle}{Proceedings of the 20th International
  Conference on Artificial Intelligence and Statistics}.
\newblock \bibinfo{journal}{\emph{Proceedings of Machine Learning Research}}
  \bibinfo{volume}{54}, \bibinfo{pages}{962--970}.
\newblock
\showeprint[arxiv]{1507.05259}


\bibitem[\protect\citeauthoryear{Zhang}{Zhang}{6 01}]%
        {zhang2016assessing}
\bibfield{author}{\bibinfo{person}{Yan Zhang}.}
  \bibinfo{year}{2016-01}\natexlab{}.
\newblock \showarticletitle{Assessing Fair Lending Risks Using Race/Ethnicity
  Proxies}.
\newblock \bibinfo{journal}{\emph{Management Science}} \bibinfo{volume}{64},
  \bibinfo{number}{1} (\bibinfo{year}{2016-01}), \bibinfo{pages}{178--197}.
\newblock
\urldef\tempurl%
\url{https://doi.org/10.1287/mnsc.2016.2579}
\showDOI{\tempurl}


\bibitem[\protect\citeauthoryear{{\v{Z}}liobait\.{e}}{{\v{Z}}liobait\.{e}}{2015}]%
        {zliobaite2015relation}
\bibfield{author}{\bibinfo{person}{Indr\.{e} {\v{Z}}liobait\.{e}}.}
  \bibinfo{year}{2015}\natexlab{}.
\newblock \showarticletitle{On the relation between accuracy and fairness in
  binary classification}.
\newblock  (\bibinfo{year}{2015}).
\newblock
\showeprint[arxiv]{1505.05723}


\end{thebibliography}

\clearpage
\begin{appendices}
  \section{Proofs for Section 3}

\begin{proof}[Proof of Theorem \ref{thm: weighted}]
By the strong law of large numbers, we have the almost sure convergence for the estimator
\begin{align*}
    \hat{\mu}_{W}(a) &=  \frac{\frac{1}{N}\sum_{i = 1}^N \pr(A_i = a \mid Z_i)Y_i}{\frac{1}{N}\sum_{i = 1}^N \pr(A_i = a \mid Z_i)} \\
    &\overset{\text{a.s.}}{\to} \frac{\expect [\pr (A = a \mid Z)Y]}{\expect [\pr (A = a \mid Z)]} \\
    &= \frac{\expect [\expect (\ind (A = a) \mid Z) \expect (Y \mid Z)]}{\pr (A = a)}.
\end{align*}
For the true mean group outcome, introduce the trivial conditional expectation with respect to $Z$:
\begin{align*}
    \mu(a) = \frac{\expect[\ind (A = a)Y]}{\pr (A = a)} = \frac{\expect [\expect (\ind (A = a)Y \mid Z)]}{\pr (A = a)}
\end{align*}
Thus, we can regroup the terms in the difference to recognize the conditional covariance
\begin{align*}
    &\hat{\mu}_{W}(a) - \mu(a) \\
    \overset{\text{a.s.}}{\to} &-\frac{\expect [\expect (\ind (A = a)Y \mid Z) - \expect (\ind (A = a) \mid Z) \expect(Y \mid Z)]}{\pr (A = a)} \\
    = & -\frac{\expect [\cov (\ind(A = a), Y \mid Z)]}{\pr (A = a)}.
\end{align*}
The analogous results hold for $A=b$ everywhere.
\end{proof}

\begin{proof}[Proof of Corollary \ref{corr:weighted-unbiased}]
When $Y$ is independent of $A$ conditionally on $Z$, it immediately follows that
\[
\cov (\ind(A = a), Y \mid Z) = \cov (\ind(A = b), Y \mid Z) = 0,
\]
and thus
\begin{align*}
    \hat{\mu}_{W}(a) - \mu(a) &\overset{\text{a.s.}}{\to} 0\\
    \hat{\mu}_{W}(b) - \mu(b) &\overset{\text{a.s.}}{\to} 0 \\
\therefore \hat{\delta}_{W} - \delta &\overset{\text{a.s.}}{\to} 0
\end{align*}
\end{proof}

\begin{proof}[Proof of Theorem \ref{thm: demographic-est-race-bias}]
Define the following events for $u_1, u_2 \in \{a, b\}$:
\begin{align*}
 \mathcal{E}_q^+(u_1, u_2) &= \{\pr (A = u_1 \mid Z) > q, A = u_2\} \\
 \mathcal{E}_q^-(u_1, u_2) &= \{ \pr (A = u_1 \mid Z) \le q, A = u_2\}.
\end{align*}
Then
\begin{align*}
    \{A = a\} &= \mathcal{E}_q^+(a, a) \cup \mathcal{E}_q^-(a, a), \\
    \{\hat{A} = a\} &= \mathcal{E}_q^+(a, a) \cup \mathcal{E}_q^+(a, b),
\end{align*}
where $\hat{A}$ is the estimated protected attribute according to the thresholding rule (Definition \ref{def:thresholding_rule}).

It follows that
\begin{align*}
    \mu(a) &= \frac{\expect [\ind (A = a)Y]}{\pr(A = a)} \\
    &= \frac{\expect [\ind (\mathcal{E}_q^+(a, a))Y] + \expect [\ind (\mathcal{E}_q^-(a, a))Y]}{\pr(\mathcal{E}_q^+(a, a)) + \pr(\mathcal{E}_q^-(a, a))} \\
    \hat{\mu}_W(a) &\overset{\text{a.s.}}{\to} \frac{\expect [\ind (\hat{A} = a)Y]}{\pr(\hat{A} = a)}  \\
    &= \frac{\expect [\ind (\mathcal{E}_q^+(a, a))Y] + \expect [\ind (\mathcal{E}_q^+(a, b))Y]}{\pr(\mathcal{E}_q^+(a, a)) + \pr(\mathcal{E}_q^+(a, b))}.
\end{align*}

Then simple algebra shows that
\begin{align*}
    \hat{\mu}_W(a) - \mu(a) &\overset{\text{a.s.}}{\to} \Delta_1(a)C_1(a) - \Delta_2(a)C_2(a) \\
    &+ (\Delta_1(a) - \Delta_2(a))C_3(a),
\end{align*}
where
\begin{align*}
    \Delta_1(a) = \expect[Y \mid \mathcal{E}_q^+(a, b)] - \expect[Y \mid \mathcal{E}_q^+(a, a)], \\
    \Delta_2(a) = \expect[Y \mid \mathcal{E}_q^-(a, a)] - \expect[Y \mid \mathcal{E}_q^+(a, a)],
\end{align*}
and
\begin{align*}
    C_1(a) = \pr(\hat{A} = a \mid A = a)\pr(A = b \mid \hat{A} = a), \\
    C_2(a) = \pr(A = a \mid \hat{A} = a)\pr(\hat{A} \ne a \mid A = a), \\
    C_3(a) = \pr(\hat{A} \ne a \mid A = a)\pr(A = b \mid \hat{A} = a).
\end{align*}
Similarly, we can prove that
\begin{align*}
    \hat{\mu}_W(b) - \mu(b) &\overset{\text{a.s.}}{\to} \Delta_1(b)C_1(b) - \Delta_2(b)C_2(b) \\
    &+ (\Delta_1(b) - \Delta_2(b))C_3(b),
\end{align*}
where $\Delta_1(b)$, $\Delta_2(b)$, $C_1(b)$, $C_2(b)$, $C_3(b)$ are defined analogously.

In summary, as ${N \to \infty}$, for $u \in \{a, b\}$ and $u^c$ as the class opposite to $u$, ,
\begin{align*}
    \hat{\mu}_{q}(u) - \mu(u) \to &\;\Delta_1(u)C_1(u) - \Delta_2(u)C_2(u) \\
    &+ (\Delta_1(u) - \Delta_2(u))C_3(u),
\end{align*}
where $\Delta_1(u)$ and $\Delta_2(u)$ are defined in Section 3.2, and

\begin{align*}
    C_1(u) &= \pr(\hat{A} = u \mid A = u)\pr(A = u^c \mid \hat{A} = u), \\
    C_2(u) &= \pr(A = u \mid \hat{A} = u)\pr(\hat{A} \ne u \mid A = u), \\
    C_3(u) &= \pr(\hat{A} \ne u \mid A = u)\pr(A = u^c \mid \hat{A} = u).
\end{align*}


\end{proof}

\begin{proof}[Proof of Corollary \ref{corollary: demographic-est-race}]
The conclusions follow immediately from Theorem \ref{thm: demographic-est-race-bias}.
\end{proof}

\begin{proof}[Proof of Corollary \ref{corollary: C-functions}]
First, (i) is obvious according to the formulas of $C_2(u), C_3(u)$ in Theorem \ref{thm: demographic-est-race-bias}.

Second, note that
\begin{align*}
C_1(u) &= \pr(\hat{A} = u \mid A = u)(1 - \pr(A = u \mid \hat{A} = u)), \\
C_2(u) &= \pr(A = u \mid \hat{A} = u)(1 - \pr(\hat{A} = u \mid A = u)).
\end{align*}

Thus
\begin{align*}
C_2(u) - C_1(u) &= \pr(A = u \mid \hat{A} = u) - \pr(\hat{A} = u \mid A = u) \\
                &= \pr(A = u, \hat{A} = u)[\frac{1}{\pr (\hat{A} = u)} - \frac{1}{\pr ({A} = u)}] \\
                &= \frac{\pr(A = u, \hat{A} = u)}{\pr (\hat{A} = u) \pr ({A} = u)}
                (\pr ({A} = u) - \pr (\hat{A} = u)).
\end{align*}

Therefore, $\pr (\hat{A} = u) < \pr ({A} = u)$ if and only if
\[
C_2(u) > C_1(u).
\]

\end{proof}

\section{Multilevel unknown protected class}
In this section, we suppose that the protected class $A$ can have more than two values and the value space of $A$ is denoted as $\mathcal{A}$. Denote $\mathcal{A}' = \mathcal{A}\backslash\{a, b\}$, i.e., $\mathcal{A}'$ is the set of all values of $A$ other than $a$ and $b$.

\begin{corollary} \label{corollary: C-functions-multiclass}
The quantities $C_1(u), C_2(u), C_3(u)$ for $u \in \{a, b\}$ in Theorem \ref{thm: demographic-est-race-bias} are related in the following way:
\begin{enumerate}[label=(\roman*)]
    \item If $\pr (A = u \mid \hat{A} = u) > \pr (A = u^c \mid \hat{A} = u)$, then $C_2(u) > C_3(u)$
    \item If $\pr (A = u) + \pr (A \in \mathcal{A}', \hat{A} = u)> \pr (\hat{A} = u)$, then $C_2(u) > C_1(u)$.
\end{enumerate}
Thus if the conditions in (i) and (ii) both hold, then $C_2(u) > C_1(u)$ and $C_2(u) > C_3(u)$.
\end{corollary}

This means in the multiclass case, $C_2(u) > C_1(u)$ can hold more easily: it can hold even if $\pr (A = u) < \pr (\hat{A} = u)$, which explains why in figure \ref{figure: C_values} $C_2(u) > C_1(u)$ even though $\pr (A = u) < \pr (\hat{A} = u)$ where $u$ stands for White.  



\section{Additional results on the mortgage dataset}
\subsection{Bias According to Theorem \ref{thm: demographic-est-race-bias}}\label{section:theory_bias}
\begin{figure}[H]
  \centering
    \includegraphics[width=\linewidth]{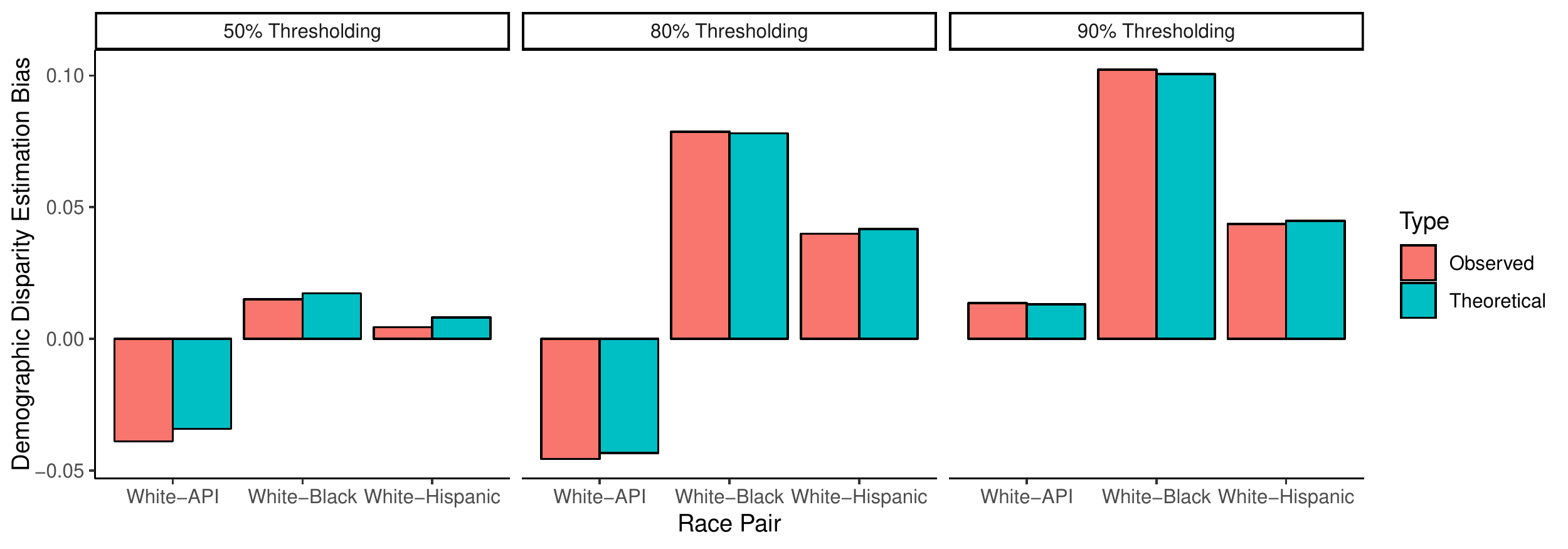}
  \caption{The observed bias is the actual estimation bias of the thresholded estimators with different thresholds, as shown in the right subfigure of Figure \ref{figure: mortgage_demo_disparity}. The theoretical bias is computed according to Theorem \ref{thm: demographic-est-race-bias}. This figure shows that the theoretical bias formula in Theorem \ref{thm: demographic-est-race-bias} for binary protected class approximates the observed bias for multiclass protected class very well. Therefore, Theorem \ref{thm: demographic-est-race-bias} indeed captures the main bias sources of thresholded estimators for both binary and multiclass protected class.} \label{figure: theory_bias}
\end{figure}

\subsection{Results for API}
\begin{figure}[H]
  \centering
    \includegraphics[width=\linewidth]{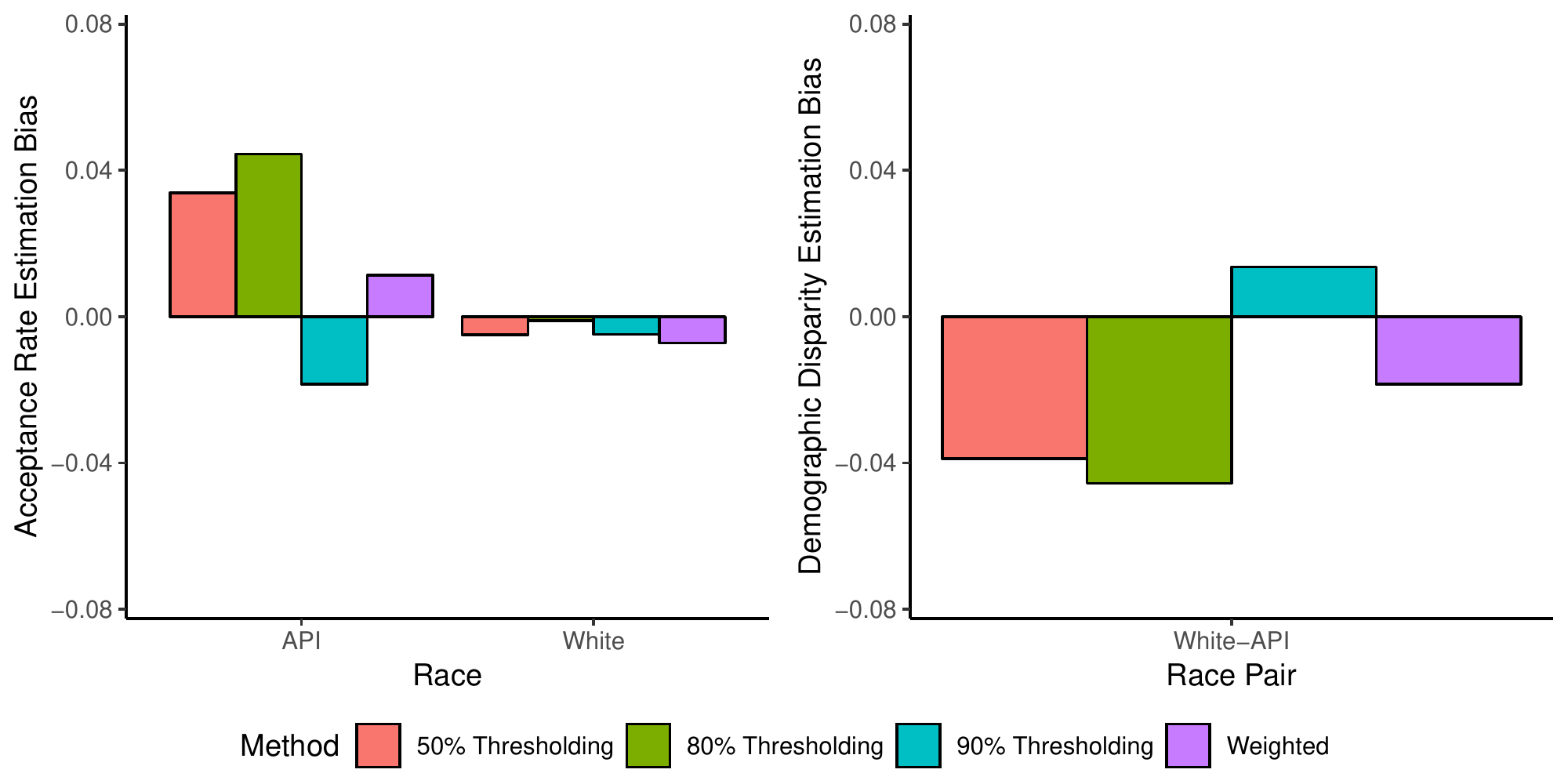}
  \caption{Left figure: the acceptance rate estimation bias for the API race. Right figure: demographic disparity estimation bias with respect to White and API.} \label{figure: mortgage_demo_api}
\end{figure}

\begin{figure}[H]
  \centering
    \includegraphics[width=\linewidth]{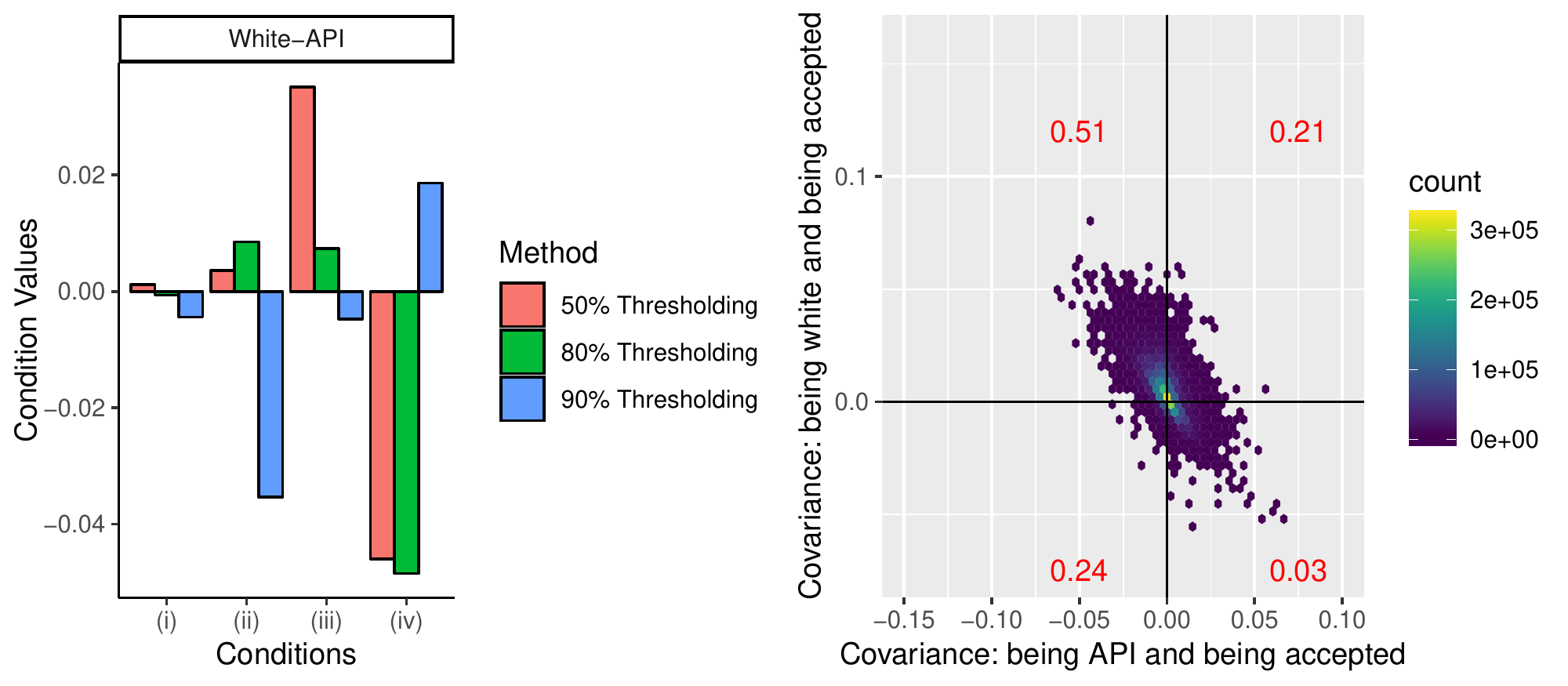}
  \caption{Left figure: the left hand side quantities in conditions \ref{eq:condition1}---\ref{eq:condition4} in Corollary \ref{corollary: demographic-est-race} when using thresholded estimator to estimate the demographic disparity with respect to White and API. Right figure: the population distribution across census tracts with different combinations of covariance terms in Theorem \ref{thm: weighted} with White as the advantaged group and API as the disadvantaged group.} \label{figure: condition_api}
\end{figure}

In Figure \ref{figure: mortgage_demo_api}, we can observe that the estimation bias is quite different when estimating the demographic disparity for White and API. The thresholded estimator could overestimate or underestimate the demographic disparity while the weighted estimator only very slightly underestimate the demographic disparity. This is not totally surprising because the API population have very different geolocation distribution and socioeconomic status distribution than Hispanic and Black. API population tend to mix with other races in terms of living locations and the average socioeconomic status disparity between API and White is much smaller than the average socioeconomic status disparity between other minority groups and White. Figure \ref{figure: condition_api} shows that the proportion of people who live in census tracts where API are accepted at lower average rate while White are accepted at higher average rate is much smaller than the proportion of people who live in census tracts where Black or Hispanic are accepted at lower average rate while White are accepted at higher average rate. This explains why the underestimation bias of the weighted estimator is very small.

\subsection{Correlation between the race probability and socioeconomic status}\label{section:correlation_income}
\begin{figure}[H]
  \centering
    \includegraphics[width=\linewidth]{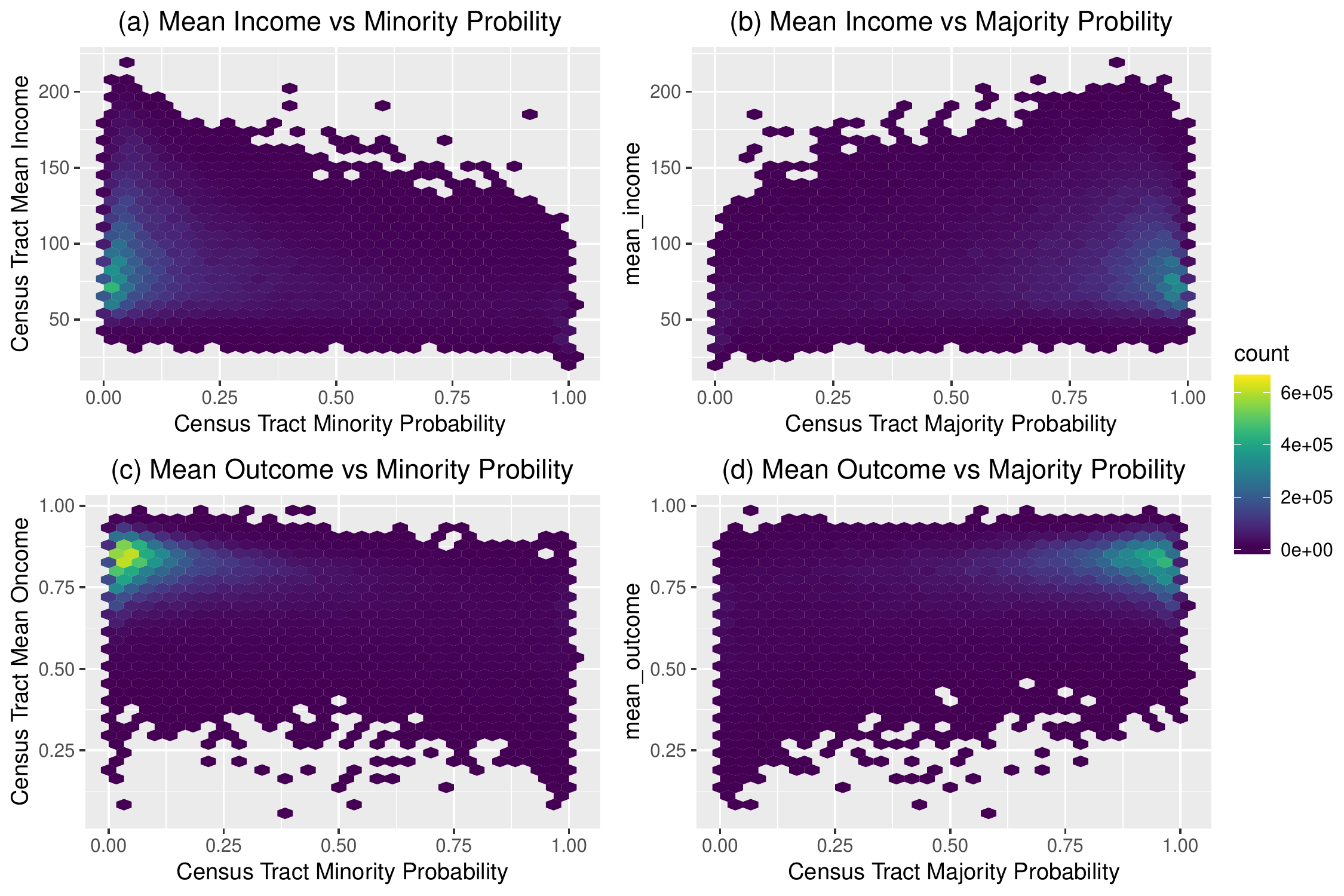}
  \caption{Figure (a) and (b) show the population distribution across census tracts with different mean yearly income versus different probability of belonging to the majority group and different probability of belonging to the minority groups. Figure (c) and (d) show the population distribution across census tracts with different average loan acceptance rate versus different probability of belonging to the majority group and different probability of belonging to the minority groups. The majority group here refers to White and the minority groups here refer to Black or Hispanic. } \label{figure: income_correlation}
\end{figure}

Figure \ref{figure: income_correlation} shows the correlation between the the race proxy probabilities and yearly income or average loan acceptance rate. We can clearly observe that census tracts with high White probability overall have more mass on higher income and higher loan acceptance rate while census tracts with high Hispanic or Black probability have more mass on lower income and lower loan acceptance rate. This figure validates the fact that the geolocation encodes socioeconomic status disparities such that the race proxy probabilities are correlated with socioeconomic status variables (e.g., income in this example) and the loan approval outcome. These correlations account for the conditions \eqref{eq:condition3}---\eqref{eq:condition4} in the Corollary \ref{corollary: demographic-est-race}.

\subsection{An example of unbiased weighted estimator} \label{section: unbiased_weighted}
\begin{figure}
  \centering
    \includegraphics[width=\linewidth]{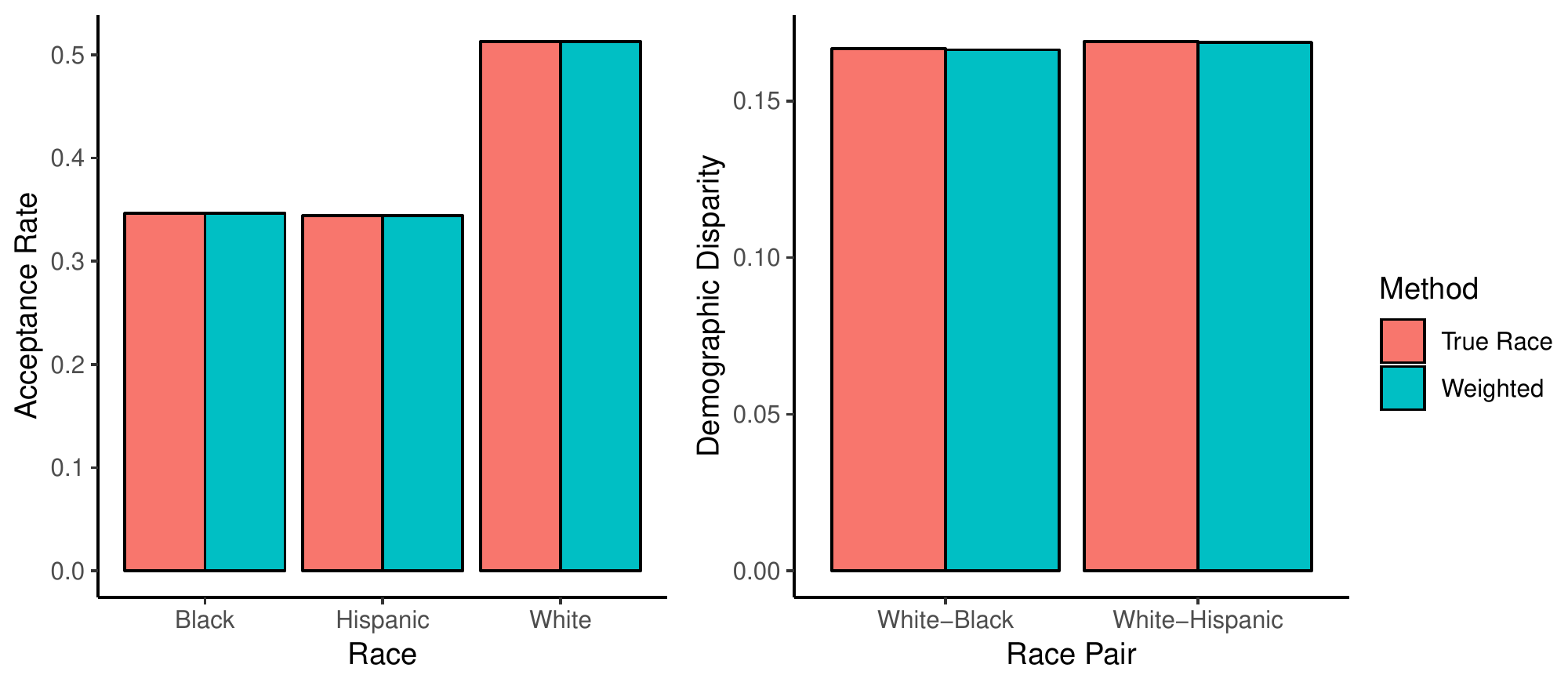}
  \caption{The loan acceptance rate and demographic disparity estimated by using the true race and by the weighted estimator for the semi-synthetic data where income is used to estimate race and it also determines the loan outcome. Weighted estimator is unbiased in this case because race is independent with the outcome conditionally on income by construction.} \label{figure: weighted_unbiased}
\end{figure}

To provide an example where the weighted estimator can be unbiased, we construct a semi-synthetic dataset based on the mortgage dataset. We aggregate the mortgage applicants' yearly income into deciles and use this discrete income as the predictor for race. We simulate the loan approval outcome $Y$ with $\pr (Y = 1 \mid Z = z) = 1/(1 + \exp(-(z - 5.5)))$, i.e., the loan approval outcome only depends on the discrete income $Z$. By construction, the loan approval outcome $Y$ is independent with race $A$ conditionally on $Z$, thus according to Theorem \ref{thm: weighted} the weighted estimator should be unbiased. This theoretical result is verified in Figure \ref{figure: weighted_unbiased}. This example shows that the weighted estimator with well constructed proxy model can estimate the demographic disparity perfectly.

\end{appendices}

\end{document}